\documentclass[conference]{IEEEtran}

\usepackage{graphics,
           psfrag,
           epsfig,
           amsthm,
           cite,
           amssymb,
           url,
           dsfont,
           subfigure,
           algorithm,
           algorithmic,
           balance,
           enumerate,
           color,
           setspace
}
\usepackage{amsmath}

\usepackage{epstopdf}

\newtheorem{definition}{Definition}

\newtheorem{theorem}{Theorem}
\newtheorem{corollary}{Corollary}

\newcommand{\eref}[1]{(\ref{#1})}
\newcommand{\sref}[1]{Section~\ref{#1}}

\newcommand{\fref}[1]{Figure~\ref{#1}}

\newcommand{\cref}[1]{Constraint~\ref{#1}}

\newcommand{\ignore}[1]{}

\IEEEoverridecommandlockouts
\begin{document}

\title{A Game-Theoretic Framework for Decentralized Cooperative Data Exchange using Network Coding}
\author{
   \authorblockN{Ahmed Douik, \textit{Student Member, IEEE}, Sameh Sorour, \textit{Member, IEEE}, Hamidou Tembine, \textit{Senior Member, IEEE},\\ Tareq Y. Al-Naffouri, \textit{Member, IEEE}, and Mohamed-Slim Alouini, \textit{Fellow, IEEE}}
   {\thanks { 
Ahmed Douik and Mohamed-Slim Alouini are with Computer, Electrical and Mathematical Sciences and Engineering (CEMSE) Division at King Abdullah University of Science and Technology (KAUST), Thuwal, Makkah Province, Saudi Arabia, email: \{ahmed.douik,slim.alouini\}@kaust.edu.sa

Sameh Sorour is with the Electrical Engineering Department, King Fahd University of Petroleum and Minerals (KFUPM), Dhahran, Eastern Province, Saudi Arabia, email:samehsorour@kfupm.edu.sa

Hamidou Tembine is with the SRI Center for Uncertainty Quantification (CEMSE) Division at King Abdullah University of Science and Technology (KAUST), Thuwal, Makkah Province, Saudi Arabia, email: tembine@ieee.org

Tareq Y. Al-Naffouri is with both the CEMSE Division at King Abdullah University of Science and Technology (KAUST), Thuwal, Makkah Province, Saudi Arabia, and the Electrial Engineering Department at King Fahd University of Petroleum and Minerals (KFUPM), Dhahran, Eastern Province, Saudi Arabia, e-mail: tareq.alnaffouri@kaust.edu.sa.

This work is an extended version of work \cite{gameconf} submitted to Globecom, Austin, Texas, USA, 2014.
}}
    }

\maketitle

\IEEEoverridecommandlockouts

\begin{abstract}
In this paper, we introduce a game theoretic framework for studying the problem of minimizing the delay of instantly decodable network coding (IDNC) for cooperative data exchange (CDE) in decentralized wireless network. In this configuration, clients cooperate with each other to recover the erased packets without a central controller. Game theory is employed herein as a tool for improving the distributed solution by overcoming the need for a central controller or additional signaling in the system. We model the session by self-interested players in a non-cooperative potential game. The utility functions are designed such that increasing individual payoff results in a collective behavior achieving both a desirable system performance in a shared network environment and the Nash bargaining solution. Three games are developed: the first aims to reduce the completion time, the second to reduce the maximum decoding delay and the third the sum decoding delay. We improve these formulations to include punishment policy upon collision occurrence and achieve the Nash bargaining solution. Through extensive simulations, our framework is tested against the best performance that could be found in the conventional point-to-multipoint (PMP) recovery process in numerous cases: first we simulate the problem with complete information. We, then, simulate with incomplete information and finally we test it in lossy feedback scenario. Numerical results show that our formulation with complete information largely outperforms the conventional PMP scheme in most situations and achieves a lower delay. They also show that the completion time formulation with incomplete information also outperforms the conventional PMP.
\end{abstract}


\section{Introduction}

\subsection{Network Coding}

Since its introduction in \cite{850663}, NC was shown to be a promising technique to significantly improve the throughput and delays of packet recovery especially in wireless erasure networks, due to the broadcast nature of their transmissions. These merits are essential for real time applications requiring reliable transmissions and fast recovery over erasure channels, such as multimedia streaming \cite{5753573}.

Two important classes of NC for such applications can be distinguished in the literature: the Random Network Coding (RNC) \cite{4015713,1228459} and the Opportunistic Network Coding (ONC) \cite{ref9,2383273}. RNC is implemented by combining packets with independent, random and non zero coefficients \cite{4015713}. Despite its attractive benefits such as optimality in number of transmissions for broadcast applications and ability to recover even without feedback \cite{1228459}, RNC is not suitable for the applications of our interest since it does not allow progressive decoding of the frame, is not optimal for multipoint to multipoint communications (multicast) and require expensive computation at the clients to decode the frame. In ONC, packet combinations are selected according to the received/lost packet state of each client \cite{2383273}. ONC was show to be a graceful solution for packet recovery for wireless network \cite{49413587}.

One ONC subclass that suits most of the aforementioned applications is the instantly decodable network coding (IDNC) since it provides instant and progressive decoding of packets. IDNC can be implemented using binary XOR to encode and decode packets. Furthermore, no buffer is needed at the clients to store non instantly decodable packets for future decoding possibilities. Thanks to its merits, IDNC was an intensive subject of research \cite{ref2,ref3,ref6,xiao1,letterarxiv,xiao2,ref12,ref18,5753573,6620795,6655395,6725590,6766433,6120247}. In \cite{ref2,ref3}, the sum decoding delay for PMP configuration was studied for perfect feedback then extended to limited feedback scenario in \cite{refsameh,refahmed,ref17,letterarxiv}. The maximum decoding delay was introduced in \cite{vtc} as a more reliable delay metric in IDNC, and algorithms to minimize this delay for PMP configuration have been proposed. In \cite{ref4}, the authors studied the problem of minimizing the completion time in IDNC. The completion time was proofed to be related to the decoding delay in \cite{confarxiv} and can be better controlled with it. 

\subsection{Motivation and Related Work}

In all aforementioned works, the base station of a point-to-multipoint network (such as cellular, Wi-Fi and WiMAX and roadside to vehicle networks) was assumed to be responsible for the recovery of erased packets. This can pose a threat on the resources of such base stations and their abilities to deliver the required huge data rates especially in future wireless standards. This problem becomes more severe in roadside to vehicle networks since the vehicles usually bypass the roadside senders very fast and thus cannot rely on it for packet recovery but rather on completing the missing packets among themselves. One alternative to this problem is the notion of cooperative data exchange (CDE) introduced in \cite{6404743}. In this configuration, clients can cooperate to exchange data by sending IDNC recovery packets to each other over short range and more reliable communication channels, thus allowing the base station to serve other clients. This CDE model is also important for fast and reliable data communications over ad-hoc networks, such vehicular and sensor networks. Consequently, it is very important to study the minimization of delays and number of transmissions in such IDNC-based CDE systems.

Unlike conventional point-to-multipoint scenario, the IDNC based CDE systems require not only decisions on packet combinations but also on which client to send in every transmission in order to achieve a certain quality for one of the network metrics. Recently Aboutorab and al. \cite{6620795} considered the problem of minimizing the sum decoding delay for CDE in a centralized fashion. By centralized, we mean that a central unit (such as the base station in the cellular example) takes the decisions on which client to send and which packet combination in each transmission.

The aforementioned work considered a perfect and prompt feedback from all the players. This assumption is too idealistic given the impairments on feedback link \cite{refjournal} of wireless networks, due to shadowing, high interference and fading. In this situation, players will need to transmit several subsequent packets without having any information (or partial information) about their reception status at the other players. Moreover, this scenario introduce a new dimension to the problem since the information at the different players is no longer common knowledge. When IDNC based CDE is used in such uncertainties and non symmetric information, the players will no longer be certain on the outcome of their actions, and thus will not be certain about the resulting completion time and decoding delays. 

\subsection{Contributions}

In this paper, we introduce a game theoretic framework for studying the problem of minimizing the delay of instantly decodable network coding for cooperative data exchange in decentralized wireless network. The problem is modeled as cooperative control problem using game theory as a tool for improving the distributed solution by overcoming the need for a central controller or additional signaling in the system.

Cooperative control problems entail numerous autonomous players seeking to collectively achieve a global objective. The network coding problem is one example of a cooperative control problems, in which the global objective is for all players to efficiently use a common resource by opportunistically taking advantage of the possible coding occasions. The central challenge in cooperative control problems is to derive a local control mechanism for the individual players such that the players operate in a manner that collectively serves the desired global objective.

In this paper, we derive expressions for the individual utility functions in such way that increasing individual payoff results in a collective behavior achieving a desirable system performance (minimizing one of the delay aspects) in a shared network environment and achieving the Nash bargaining solution. We then improve these game formulations to include punishment policy and reduce the set of equilibrium to the one dimensional line containing the Nash bargaining solution of our interest. To the best of our knowledge, using game theory tools to model IDNC-based CDE has not been addressed in the literature and only heuristic algorithm were proposed to solve the problem in a centralized fashion \cite{6620795}. Moreover, this work can serve as a background to build more complicated system such as the multicast in which the packet demand of each player can differ and players are not all in the transmission range of each other.

The rest of this paper is divided as follows: Background about game theory and specially potential games is briefly recalled in \sref{sec:back}. In \sref{sec:model}, we present our network model and protocol. The game parameters, formulations and equilibrium investigation are presented in \sref{sec:game1}. The punishment policy and the new game formulations with their equilibrium are provided in \sref{sec:game2}. In \sref{sec:algo}, we present the algorithm used to simulate the system. \sref{sec:ext} presents an extension of our study to the lossy feedback scenario. Before concluding in \sref{sec:concl}, simulations results are illustrated in \sref{sec:simul}.

\section{Background: Non-cooperative Games}\label{sec:back}

\subsection{Definitions and Notations}

We define a finite stochastic non-cooperative game, like in \cite{Lasaulce:2011:GTL:2086746}, with a common state by the 6-uplet:
\begin{align}
\mathcal{G}= (\mathcal{M},\{\overline{\mathcal{A}}_i\}_i,\Omega,\{\mathcal{A}_i\}_i,\{\alpha_i\}_i,q,\{\mathcal{U}_i\}_i) ,
\end{align} 
where:
\begin{itemize}
\item $\mathcal{M}=\{1,...,M\}$ is the set of players,
\item $\{\overline{\mathcal{A}}_i\}_i$ is the set of all possible actions during the course of the game,
\item $\Omega$ is the set of possible states of the game,
\item $\mathcal{A}_i(\omega)$ is the set of possible action for player $i$ in the state $w \in \Omega$ of the game,
\item $\alpha_i: \Omega \longrightarrow 2^{\overline{\mathcal{A}}_i}$ is the correspondence determining the possible actions at a given state of a game,
\item $q_t$ is the conditional distribution of the transition probability from state to state. For independent states games, $q$ is the distribution over the set $\Omega$ and can be ignored in the definition of the game. Otherwise the game is called competitive Markov decision process \cite{filar-competitiveMDP},
\item $\mathcal{U}_i$ is the utility function of player $i$, which will be defined further in the paper.
\end{itemize}

Let $\omega(t) \in \Omega$ be the state of the game at the stage $t$. For notation simplicity, the set of actions of player $i$ at stage $t$ will be denoted by $\mathcal{A}_i(t)$ instead of $\mathcal{A}_i(\omega(t))$. Let $\mathcal{A}(t) = \mathcal{A}_1(t) \times ... \times \mathcal{A}_M(t)$ be the set of all possible actions that can be taken by all the players at the stage $t$ of the game. 

For an action profile $\underline{a}_t=(a_1(t),...,a_M(t))^{\textbf{T}} \in \mathcal{A}(t)$, let $\underline{a}_{t,-i}$ denote the profile of players other than player $i$. In other words, $\underline{a}_{t,-i}=(a_1(t),...,a_{i-1}(t),a_{i+1}(t),...a_M(t))^{\textbf{T}}$. The subscript $^{\textbf{T}}$ denote the transpose operator. We can write a profile $\underline{a}_t$ of actions as $(\underline{a}_{t,i},\underline{a}_{t,-i})$. Similarly, the notation $\mathcal{A}_{-i}(t) = \prod\limits_{j \neq i} \mathcal{A}_j(t)$ refers to the set of possible actions of all the player other than player $i$ at stage $t$ of the game. Let $\underline{h}_t = (\boldsymbol{\omega}(1),\underline{a}(1),...,\underline{a}(t-1),\boldsymbol{\omega}(t) )^{\textbf{T}}$ be the history of the game at stage $t$ that lies in the set:
\begin{align}
\mathcal{H}_t = \left( \bigotimes_{t^\prime = 1}^{t-1} \Omega \times \mathcal{A}(t^\prime)  \right) \times \Omega.
\end{align}
The utility function $\mathcal{U}_i$ for player $i$ is defined as:
\begin{align}
\begin{tabular}{c|ccc}
$\mathcal{U}_i:$ & $\mathcal{A}(t) \times \mathcal{H}_t$ & $\longrightarrow$ & $\mathds{R}$  \\
 & $ (\underline{a}_t,\underline{h}_t)$ & $\longmapsto$ & $ \mathcal{U}_i(\underline{a}_t,\underline{h}_t),$
\end{tabular}
\end{align}
where in the notation $X \longrightarrow Y$, $X$ refers to the set of arguments and $Y$ the set of images by the function and $x \longmapsto f(x)$ gives the mapping of each argument.

We may write $\mathcal{U}_i(\underline{a}_t,\underline{h}_t)$ as $\mathcal{U}_i(\underline{a}_{t,i},\underline{a}_{t,-i},\underline{h}_t)$. For clarity purposes, the notation $\textbf{X}$ refers to a matrix whose $i$th column  is $\underline{X}_i$. The notation $\underline{x}$ refers to a vector whose $i$th entry is $x_i$. We denote by $\{0,1\}^{x\times y}$ the set of matrices of dimension $x \times y$ containing only $0$s and $1$s. We also use the notation $\{0,1\}^{x}$ to refer to the set $\{0,1\}^{x\times 1}$. The notation $[\underline{X}_1,...,\underline{X}_n]$ refers to the matrix whose $i$th column is the vector $\underline{X}_i$ and the notation $\textbf{X}=[x_{ij}]$ refers to a matrix $\textbf{X}$ whose $i$th row and $j$ column is the element $x_{ij}$. The game $\mathcal{G}$ is said to be finite if the number of times it is played is finite. For such games, let $T$ be the final stage of the game.

\subsection{Potential Games}

Potential games have been introduced in \cite{Monderer1996124}. The definition of a potential game $\mathcal{G}$ is given by:
\begin{definition}
The game $\mathcal{G}$ is an \emph{exact} potential game if there exists a function $\phi$ such that:
\begin{align}
&\hspace{2cm}\forall~ i \in \mathcal{M}, \forall~ t , \forall~ \underline{a}_t, \underline{a}_t^\prime \in \mathcal{A}(t) ,  \\
&\mathcal{U}_i(\underline{a}_t,\underline{h}_t)-\mathcal{U}_i(\underline{a}_{t,i}^\prime,\underline{a}_{t,-i},\underline{h}_t) = \nonumber \\
& \hspace{3.5cm}\phi(\underline{a}_t,\underline{h}_t)-\phi(\underline{a}_{t,i}^\prime,\underline{a}_{t,-i},\underline{h}_t).\nonumber
\end{align}
In other words, a game $\mathcal{G}$ is an \emph{exact} potential game if there is a function $\phi$ that measures exactly the difference in the utility due to unilaterally deviation of each player \cite{Lasaulce:2011:GTL:2086746}.
\end{definition}

Such a function $\phi$ is called the \emph{exact} potential of the game. Note that such potential does not directly guarantee the \textit{Pareto optimality} of the \textit{Nash Equilibrium} (see Cournot oligopolies \cite{Monderer1996124}). Both \textit{Pareto optimality} and \textit{Nash Equilibrium} will be defined in the next section. Instead of being a warranty of Pareto efficiency, the potential function can be seen as a quantification of the disagreement among the players \cite{Monderer1996124}. In the dynamic system \cite{1660950}, the potential represents a Lyapunov function of the game.

Note that other type of potential games can be found in the literature: the weighted, ordinal, generalized and best-response potential games \cite{Voorneveld2000289}. The theorems stated in the next section will not depend on the nature of the potential game considered and hold for all of them \cite{Lasaulce:2011:GTL:2086746} as far as the game is potential.

\subsection{Equilibrium Existence and Pareto Optimality}

The definition of the Pure Nash Equilibrium (NE) is the following:
\begin{definition}
An action profile $\underline{a}_t^* \in \mathcal{A}(t)$ is called a Pure Nash equilibrium if:
\begin{align}
\forall~i \in \mathcal{M},~ \mathcal{U}_i(\underline{a}_{t}^*,\underline{h}_t) = \underset{\underline{a}_{t,i} \in \mathcal{A}_i(t)}{\text{max}} \mathcal{U}_i(\underline{a}_{t,i},\underline{a}_{t,-i}^*,\underline{h}_t).
\end{align}
In other words, a NE is an action profile $\underline{a}_t^*$ in which no player can increase his utility by unilateral deviation. In engineering system, the NE is a stable point to operate \cite{jstor}.
\end{definition}

An attractive property of the NE is called the Pareto-Optimum Nash Equilibrium (PONE) which is defined as:
\begin{definition}
An action profile $\underline{a}_t^* \in \mathcal{A}(t)$ is called a Pareto-Optimum Nash Equilibrium if for all NE action profile $\underline{b}_{t}^*$, we have:
\begin{align}
\forall~i \in \mathcal{M},~ \mathcal{U}_i(\underline{a}_{t}^*,\underline{h}_t) \geq \mathcal{U}_i(\underline{b}_{t}^*,\underline{h}_t).
\end{align}
In other words, the PONE is the NE that achieves the highest utility for all the players among all the other NE.
\end{definition}

General results about equilibrium existence and uniqueness are provided in \cite{1965}. Since our problem is a cooperative control game, then we can make use of the results of the potential games \cite{4814554}. Before stating the main results concerning equilibrium of potential games, we first define the coordination game \cite{1660950}:
\begin{definition}
Let $\mathcal{G}= (\mathcal{M},\{\overline{\mathcal{A}}_i\}_i,\Omega,\{\mathcal{A}_i\}_i,\{\alpha_i\}_i,q,\{\mathcal{U}_i\}_i)$ be a potential game with a potential function $\phi$. The game $\mathcal{G}^\prime= (\mathcal{M},\{\overline{\mathcal{A}}_i\}_i,\Omega,\{\mathcal{A}_i\}_i,\{\alpha_i\}_i,q,\phi)$ is called the coordination game of $\mathcal{G}$.
\end{definition}
The following theorem gives the relationship in terms of equilibrium between the potential game and its associated coordination game:
\begin{theorem}
Let $\mathcal{G}$ be a potential game with potential $\phi$ and $\mathcal{G}^\prime$ its associated coordination game. Then the set of NEs of $\mathcal{G}$ coincides with the set of NEs of $\mathcal{G}^\prime$. Moreover, the actions profile $\underline{a}_t^* \in \mathcal{A}(t)$ maxima of $\phi$ are NE of $\mathcal{G}$.
\label{th1}
\end{theorem}
\begin{proof}
The proof of this theorem can be found in \cite{1660950}.
\end{proof}
Note that the converse of this theorem is not generally true i.e. not all the NEs of $\mathcal{G}$ are maxima of $\phi$ \cite{1660950}. The existence of a NE in potential game is ensured by this corollary:
\begin{corollary}
Every finite potential game admits at least one NE. 
\label{cor}
\end{corollary}
\begin{proof}
The proof comes directly from Theorem \ref{th1}. For a finite game, the potential function is finite and therefore have at least one maximum and thus the game admits at least one NE.
\end{proof}
The following theorem characterize the PONE in a coordination game:
\begin{theorem}
Let $\mathcal{G}$ be a potential game with potential $\phi$ such that its corresponding coordination game is itself i.e. $\mathcal{G}=\mathcal{G}^\prime$ then the maximum of $\phi$ is the PONE of $\mathcal{G}$.
\label{th2}
\end{theorem}
\begin{proof}
According to Theorem \ref{th1}, the maximum $\underline{a}_t^*$ of $\phi$ is a NE of $\mathcal{G}$ and since $\phi$ is the utility function of $\mathcal{G}$ therefore $\underline{a}_t^*$ is a NE that yields the highest utility. More specifically $\underline{a}_t^*$ yields the highest utility among all the NE of $\mathcal{G}$ and thus it is the PONE.
\end{proof}

\section{Network Model and Protocol} \label{sec:model}

\subsection{Network Model}

The network we consider in this paper consists of a set $\mathcal{M}=\{1,...,M\}$ geographically close clients (players) that require the reception of source packets that the base station (BS) holds. Each player is interested in receiving the frame $\mathcal{N}=\{1,...,N\}$ of source packets regardless of the order. 

In the first $n$ time slots, the BS broadcasts the $N$ source packets of the frame $\mathcal{N}$ uncoded. Each player $i$ is experiencing a packet erasure probability $q_i$ assumed to be constant during this phase. Each player listens to the transmitted packets and sends an acknowledgement (ACK) upon each successful reception of each packet. We assume that at the end of this \emph{initial} phase, each packet of the frame is at least acknowledged by one of the players. Otherwise, this packet is re-transmitted by the BS.

After this initialization phase, for each player $i$, the packets of the frame $\mathcal{N}$ can be in one of the following sets:
\begin{itemize}
\item The Has set (denoted by $\mathcal{H}_i$): The sets of packets successfully received by player $i$.
\item The Wants set (denoted by $\mathcal{W}_i$): The sets of packets that were erased at player $i$. Clearly, we have $\mathcal{W}_i = \mathcal{N} \setminus \mathcal{H}_i$.
\end{itemize}

In this configuration, we assume a perfect reception of the acknowledgement by all the players and that each player knows the packets sets of all the other players. Each player stores the information obtained after the transmission at time $(t-1)$ in a \emph{state matrix} (SM) $\textbf{S}(t) = [s_{ij}(t)],~ \forall~ i \in \mathcal{M},~ \forall~j \in \mathcal{N}$ such that:
\begin{align}
s_{ij}(t) =
\begin{cases}
0 \hspace{0.9 cm}& \text{if } j \in \mathcal{H}_i(t) \\
1 \hspace{0.9 cm}& \text{if } j \in \mathcal{W}_i(t)  .
\end{cases}
\end{align}

Note that the system model presented at the end of the \emph{initial phase} can be seen as a wireless sensor network that need to exchange a set $\mathcal{N}$ of packets and each node holds a subset (maybe overlapping) of the frame. An application of in wireless sensor network is presented in \sref{sec:algo}. In \sref{sec:ext}, we will extend our study to the limited feedback scenario.

\subsection{Network Protocol}

After the initial transmission, the recovery phase starts. In this phase, the players cooperate to recover their missing packets by transmitting to each other binary XOR encoded packets of the source packets they already hold in order to minimize the decoding delay. 

The packet combination is chosen according to the available packets they have, the information available in the SM and the expected erasure patterns of the links. Let $\textbf{P} = [p_{ij}]$, $i,j\in\mathcal{M}$ denote the packet erasure probability (i.e. the probability to loss a packet) from player $j$ to player $i$. All the packet erasure probabilities are assumed to be constant during the transmission of the frame. Since the packet erasure probability depends not only on the link but also on the available power used to transmit, therefore $p_{ij}$ can be different from $p_{ji}$. We assume that each player knows all the packet erasure probabilities linking him to other players (i.e. player $i$ knows $p_{ji},~\forall~ j \in \mathcal{N}$) and that each transmission can be heard by all the players. Therefore only one player will transmit a packet combination at each time slot. Otherwise, due to interference between transmissions, none of the players will be able to decode a packet.

In this phase, the transmitted coded packets can be one of the following three options for each player $i$:
\begin{itemize}
\item \emph{Instantly Decodable:} A packet is instantly decodable for player $i$ if the encoded packet contain at most one packet the player does not have so far. In other words, it contains \emph{only one packet} from $\mathcal{W}_i$.
\item \emph{Non-Instantly Decodable:} A packet is non instantly decodable for player $i$ if it contains more than one packet missing for that player. In other words, it contains at least two packets from $\mathcal{W}_i$.
\item \emph{Non-innovative:} A packet is non-innovative for player $i$ if it do not allow him to reduce its Wants set. In other words, it does not contains packets from $\mathcal{W}_i$.
\end{itemize}
We define the conventional decoding delay \cite{ref2,ref3} as follows:
\begin{definition}
At any cooperative phase transmission, a player $i$, with non-empty Wants set, experiences a one unit increase of decoding delay if it successfully receivers a packet that is either non-innovative or non-instantly decodable.
\end{definition}
The cooperation decoding delay can be defined as:
\begin{definition}
At any cooperative phase transmission, a player $i$, with non-empty Wants set, experiences a one unit increase of decoding delay if not exactly one player transmitted or only one player transmitted and its conventional decoding delay increases.
\end{definition}

In other words, if more than one or none players transmits, all the players will experience a decoding delay and if player $i$ is the only transmitting player, he will experience a delay along with all players that successfully received a packet that is either non-innovative or non-instantly decodable. In the rest of the paper, we will use the term decoding delay to refer to the cooperative decoding delay. We define the targeted players by a transmission as the players that can instantly decode a packet from that transmission. After each transmission in the recovery phase, its targeted players send acknowledgements consisting of one bit indicating the successful reception. This process is repeated until all players report that they obtained all the packets. Let $T$ be final stage of the game. The definition of $T$ can be written as:
\begin{align}
T = \text{min }\{ t \in \mathds{N}^* \text{ such that } \textbf{S}(t) = \textbf{0} \}.
\end{align}

Since we assume single hop transmissions, which means that all the players are in the transmission range of each other, each of them can already overhear all the feedback sent by the other players and thus the system does not require any additional feedback load.

\section{Game Formulation}\label{sec:game1}

In this section, we first introduce the game parameter to be able to model the problem of minimizing the delay in IDNC as a non-cooperative potential game. We then provide the expression of the utilities for the completion time, the sum decoding delay and the maximum decoding delay games. As first formulation of the problem, we will consider the delay as the natural cost in our games. We finally analyze the different NE equilibrium present in these games. 

\subsection{Game Parameters}

Let $\underline{\kappa}^i(t)$ be the optimal packet combination that player $i$ can generate at the stage $t$ of the game. We have $\underline{\kappa}^i(t) \in \{0,1\}^{N}$ with $\kappa^i_j(t)=1$ means that packet $j$ is included in the packet combination and $0$ otherwise. The mathematical expression of these combination can be found in \cite{jsaccompletion} for the completion time, in \cite{vtc} for the maximum decoding delay and in \cite{ref2} for the sum decoding delay. Since the SM is known by all players, therefore each player can compute the optimal packet combination (or a sub-optimal since the computation of the optimal was shown to be NP-hard) of all the other players.

At each stage of the game, each player has two possible actions: either he transmits or he listens. Therefore, we define the action space of player $i$ at each stage $t$ of the game as $\mathcal{A}_i(t) = \{ \text{transmit }\underline{\kappa}^i(t) , \text{remain silent}\}$. Note that the action space of the players are not symmetric since each player can transmit a different packet combination at each stage. Let $\underline{a}_t$ be the actions taken by all the users at the stage $t$. For simplicity of notation, we will define $\underline{a}_t$ as the following:
\begin{align}
\begin{tabular}{c|ccc}
$\underline{a}_t:$ & $\mathcal{A}(t)$ & $\longrightarrow$ & $\{0,1\}^{M}$  \\
 & $(a_1(t),...,a_M(t))$ & $\longmapsto$ & $ \underline{a}_t = (b_1(t),...,b_M(t))^{\textbf{T}},$
\end{tabular}
\end{align}
where,
\begin{align}
b_i(t) = 
\begin{cases}
0 \hspace{0.5cm}& \text{if } a_i(t) = \{\text{remain silent}\} \\
1 \hspace{0.5cm}& \text{otherwise}.
\end{cases}
\end{align}

The set of targeted players by a packet combination are those that can instantly decode an innovative packet from the combination. Let $\underline{\tau}_{\underline{\kappa}(t)}$ be the set of targeted player by the packet combination $\underline{\kappa}$ at the stage $t$ of the game. The mathematical definition of this set is given by:
\begin{align}
\begin{tabular}{c|ccc}
$\underline{\tau}:$ & $\{0,1\}^{N}$ & $\longrightarrow$ & $\{0,1\}^{M}$\\
 & $\underline{\kappa}(t)$ & $\longmapsto$ & $\underline{\tau}_{\underline{\kappa}(t)} = (\tau_1(\underline{\kappa}(t)),...,\tau_M(\underline{\kappa}(t))$,
\end{tabular}
\end{align}
where:
\begin{align}
\tau_i(\underline{\kappa}(t)) = 
\begin{cases}
1 \hspace{0.5cm}& \text{if } \sum_{j=1}^N s_{ij}(t)\kappa_j(t)=1 \\
0 \hspace{0.5cm}& \text{otherwise}.
\end{cases}
\end{align}

The players that experience a conventional decoding delay after a transmission are those with non-empty Wants set and are not targeted by the transmission. Define $\underline{\overline{\tau}}_{\underline{\kappa}(t)} = 1 - \underline{\tau}_{\underline{\kappa}(t)}$ as the set of non targeted players and let $\underline{M}^w(t)$ be the set of players with non-empty Wants set defined as follows:
\begin{align}
\underline{M}^w_i(t) = 
\begin{cases}
0 \hspace{0.5cm}& \text{if } \sum_{j=1}^N s_{ij}(t)=0 \\
1 \hspace{0.5cm}& \text{otherwise}.
\end{cases}
\end{align}
The state $\boldsymbol{\omega}$ of the game is the erasure patterns of the links between each couple of players. These states can be described by the following formula:
\begin{align}
\begin{tabular}{c|ccc}
$\boldsymbol{\omega}:$ & $\mathds{N}^*$ & $\longrightarrow$ & $\Omega = \{0,1\}^{M \times M}$  \\
 &  $t$ & $\longmapsto$ & $\boldsymbol{\omega}(t) = [X_{ij}], 1 \leq i,j \leq M$,
\end{tabular}
\end{align}
where $X_{ij}$ is a Bernoulli random variable defined as follows:
\begin{align}
X_{ij} = 
\begin{cases}
0 \hspace{0.5cm}& \text{with probability } p_{ij} \\
1 \hspace{0.5cm}& \text{with probability } 1-p_{ij}.
\end{cases}
\end{align}

The $X_{ij}$ are independent of each other and therefore the $\boldsymbol{\omega}(t)$ are independent identically distributed (iid). The game can be seen as a random matrix game. We now can compute the decoding delay $\underline{\mathcal{D}}_{w,\kappa}$ experienced by all the users when user $i$ sends the packet combination $\kappa$ at the stage $t$ using the following expression:
\begin{align}
\begin{tabular}{c|ccc}
$\underline{\mathcal{D}}_{w_i,\kappa}(t):$ & $\{0,1\}^{M+N}$ & $\longrightarrow$ & $\{0,1\}^{M}$  \\
 & $(\underline{\omega}_i(t),\underline{\kappa}(t))$ & $\longmapsto$ & $\underline{\omega}_i(t) \circ \underline{\overline{\tau}}_{\underline{\kappa}(t)} \circ \underline{M}^w(t)$,
\end{tabular}
\label{qwer}
\end{align}
where $ \circ $ is the Hadamard product. We also define the cumulative decoding delay experienced by all players since the beginning of the recovery phase (beginning of the game) until the stage $t$ of the game:
\begin{align}
\begin{tabular}{c|ccc}
$\underline{\mathds{D}}:$ & $\{0,1\}^{M} \times \mathcal{H}_t$ & $\longrightarrow$ & $\{0,1\}^{M}$  \\
 & $ (\underline{a}_t,\underline{h}_t)$ & $\longmapsto$ & $ \underline{\mathds{D}}(\underline{a}_t,\underline{h}_t) $,
\end{tabular}
\end{align}
where:
\begin{align}
&\underline{\mathds{D}}(\underline{a}_t,\underline{h}_t) = \nonumber \\
&\begin{cases}
\underline{\mathds{D}}(\underline{a}_{t-1},\underline{h}_{t-1}) + \underline{M}^w(t) \hspace{0.2cm}& \text{if } ||\underline{a}_t||_1 \neq 1\\
\underline{\mathds{D}}(\underline{a}_{t-1},\underline{h}_{t-1}) + \underline{\mathcal{D}}_{w_i,\kappa^i}(t) \hspace{0.2cm}& \text{otherwise } \\
 \text{with } i \text{ such that } a_i(t) = ||\underline{a}_t||_1. &
\end{cases}
\end{align}

In the case of the Point to Multipoint (PMP) recovery process (when only the base station is transmitting), the completion time $\mathcal{C}_i$ can be approximated \cite{jsaccompletion} by the following expression:
\begin{align}
\tilde{\mathcal{C}}_i = \cfrac{\mathcal{W}_i + \mathds{D}_i - q_i}{1-q_i}.
\end{align}
Since in the CDE, all users may be transmitting to each other, then the completion time can be approximated by:
\begin{align}
\mathcal{C}_i = \cfrac{\mathcal{W}_i + \mathds{D}_i - \overline{p}_i}{1-\overline{p}_i},
\end{align}
where $\overline{p}_i$ is the average erasure probability linking the player $i$ to the other players. This average erasure probability can be expressed in terms of the erasure matrix as follows: $\overline{p}_i= \cfrac{||\underline{P}_i||_1}{M}$. Let $\underline{\mathcal{C}}$ be the vector of the completion times, $\underline{\mathcal{W}}$ the vector of the size of the Wants sets and $\underline{\overline{p}}$ the one of the average erasures. Therefore, we define the expected completion time of each player as follows:
\begin{align}
\begin{tabular}{c|ccc}
$\underline{\mathcal{C}}:$ & $\{0,1\}^{M} \times \mathcal{H}_t$ & $\longrightarrow$ & $\{0,1\}^{M}$  \\
& $(\underline{a}_t,\underline{h}_t)$&$\longmapsto$&$(\underline{\mathcal{W}} + \underline{\mathds{D}}(\underline{a}_t,\underline{h}_t) - \underline{\overline{p}})./(1-\underline{\overline{p}}) $,
\end{tabular} 
\end{align}
where the operator $\underline{x}./\underline{y}$ refers to the division of each element of vector $\underline{x}$ by the element of vector $\underline{y}$. 

\subsection{Utility Functions}

In this section, we give the formulation of the completion time, maximum decoding delay and sum decoding delay games. In this first formulation, we take the cost (-utility) function to be the natural delay in the network.

\textbf{Game 1 (Completion time Game):}
\begin{itemize}
\item \textit{Players}: Users in set $\mathcal{M}$
\item \textit{History }: $\underline{h}_t=$ Channel realization $\boldsymbol{\omega}(t)$ and players' action $\underline{a}_t$ at each stage $t \geq 1$.
\item \textit{Strategies}: Contingency plans for selection transmission policy at each stage $t \geq 1$ and for any given history $\underline{h}_t$. 
\item \textit{Utilities}: $\mathcal{U}_i^{CT}$ for each player $i$, where at each stage $t \geq 1$ and for any given history $\underline{h}_t$ and action profile $\underline{a}_t$:
\end{itemize}
\begin{align}
\begin{tabular}{c|ccc}
$\mathcal{U}_i^{CT}:$ & $\{0,1\}^{M} \times \mathcal{H}_t$ & $\longrightarrow$ & $\mathds{R}$  \\
 & $  (\underline{a}_t,\underline{h}_t)$ & $\longmapsto$ & $ -||\underline{\mathcal{C}}(\underline{a}_t,\underline{h}_t)||_\infty$.
\end{tabular}
\end{align}
Game 1 is a non-cooperative stochastic game \cite{Lasaulce:2011:GTL:2086746}. Since the utility function do not depend on the player, therefore this game belongs to the potential game class.

\textbf{Game 2 (Maximum decoding delay Game):}
\begin{itemize}
\item \textit{Players}: Users in set $\mathcal{M}$
\item \textit{History }: $\underline{h}_t=$ Channel realization $\boldsymbol{\omega}(t)$ and players' action $\underline{a}_t$ at each stage $t \geq 1$.
\item \textit{Strategies}: Contingency plans for selection transmission policy at each stage $t \geq 1$ and for any given history $\underline{h}_t$. 
\item \textit{Utilities}: $\mathcal{U}_i^{MDD}$ for each player $i$, where at each stage $t \geq 1$ and for any given history $\underline{h}_t$ and action profile $\underline{a}_t$:
\end{itemize}
\begin{align}
\begin{tabular}{c|ccc}
$\mathcal{U}_i^{MDD}:$ & $\{0,1\}^{M} \times \mathcal{H}_t$ & $\longrightarrow$ & $\mathds{R}$  \\
 & $  (\underline{a}_t,\underline{h}_t)$ & $\longmapsto$ & $ -||\underline{\mathds{D}}(\underline{a}_t,\underline{h}_t)||_\infty$.
\end{tabular}
\end{align}
Similar to Game 1, Game 2 is also a non-cooperative stochastic potential game.

\textbf{Game 3 (Sum decoding delay Game):}
\begin{itemize}
\item \textit{Players}: Users in set $\mathcal{M}$
\item \textit{History }: $\underline{h}_t=$ Channel realization $\boldsymbol{\omega}(t)$ and players' action $\underline{a}_t$ at each stage $t \geq 1$.
\item \textit{Strategies}: Contingency plans for selection transmission policy at each stage $t \geq 1$ and for any given history $\underline{h}_t$. 
\item \textit{Utilities}: $\mathcal{U}_i^{SDD}$ for each player $i$, where at each stage $t \geq 1$ and for any given history $\underline{h}_t$ and action profile $\underline{a}_t$:
\end{itemize}
\begin{align}
\begin{tabular}{c|ccc}
$\mathcal{U}_i^{SDD}:$ & $\{0,1\}^{M} \times \mathcal{H}_t$ & $\longrightarrow$ & $\mathds{R}$  \\
 & $  (\underline{a}_t,\underline{h}_t)$ & $\longmapsto$ & $ -||\underline{\mathds{D}}(\underline{a}_t,\underline{h}_t)||_1$.
\end{tabular}
\end{align}
Similar to Game 1 and Game 2, Game 3 is also a non-cooperative stochastic potential game.

\subsection{Game Equilibrium Analysis}

The following Theorem gives the set of the NE of Game 1:
\begin{theorem}
The set of NE of Game 1 is:
\begin{align}
E(t) =  
\begin{cases}
\mathcal{A}(t) \text{ if } Z(t) = \varnothing \\
E_1(t)  \text{ otherwise },
\end{cases} 
\end{align}
with
\begin{align}
&E_1(t) = \{ \underline{a}_{t} \in \mathcal{A}(t)  \text{ such that } ||\underline{a}_{t}||_1=1 \text{ or } ||\underline{a}_{t}||_1>2  \nonumber\\
& \qquad \text{ or } (||\underline{a}_{t}||_1=a_i(t)+a_j(t)=2 \text{ and } i,j \notin Z(t))\}.
\end{align}
where the set $Z$ is the set defined by
\begin{align}
Z(t) = \{ j \in \mathcal{M} \text{ such that }  Y_j(t) < Y_0(t) \},
\end{align}
and $Y_0(t) = \underset{i \in Q(t)}{max }\cfrac{1}{1-\overline{p}_i}$ is the increase in the cost function when not exactly one player is transmitting and $Y_j(t) = \underset{i \in Q(t) \cap \underline{\mathcal{D}}_{w_j,\kappa^j}(t)}{max }\cfrac{1}{1-\overline{p}_i}$ the cost when only player $j$ is transmitting with $Q(t)$ defined as:
\begin{align}
&Q(t) = \{ i \in \mathcal{M} \text{ such that }  \\
& \qquad \mathcal{C}_i(\underline{a}_{t-1},\underline{h}_{t-1}) + 1/(1-\overline{p}_i) > ||\underline{\mathcal{C}}(\underline{a}_{t-1},\underline{h}_{t-1})||_\infty \nonumber \\
& \hspace{2cm} \text{ and } M^w_i = 1\}. \nonumber
\end{align} 
\end{theorem}
\begin{proof}
The proof of this Theorem can be found in Appendix A.
\end{proof}
The set of NE of Game 2 are given in the following Theorem:
\begin{theorem}
The set of NE of Game 2 is:
\begin{align}
E(t) =  
\begin{cases}
\mathcal{A}(t) \text{ if } Z(t) = \varnothing \\
E_1(t)  \text{ otherwise },
\end{cases} 
\end{align}
with
\begin{align}
&E_1(t) = \{ \underline{a}_{t} \in \mathcal{A}(t)  \text{ such that } ||\underline{a}_{t}||_1=1 \text{ or } ||\underline{a}_{t}||_1>2  \nonumber\\
& \qquad \text{ or } (||\underline{a}_{t}||_1=a_i(t)+a_j(t)=2 \text{ and } i,j \notin Z(t))\}.
\end{align}
where the set $Z$ is the set defined by
\begin{align}
Z(t) = \{i \in \mathcal{M} \text{ such that } \underline{\mathcal{D}}_{\omega_i,\kappa^i} \circ \underline{q}_t = \underline{0 }\},
\end{align}
with $q_i(t)=1$ if $\mathds{D}_i(\underline{a}_{t-1},\underline{h}_{t-1}) = ||\underline{\mathds{D}}(\underline{a}_{t-1},\underline{h}_{t-1})||_\infty$ and $M^w_i = 1$, otherwise $0$.
\end{theorem}
\begin{proof}
The proof of this Theorem can be found in Appendix B.
\end{proof}
The following Theorem introduces the set of NE of Game 3:
\begin{theorem}
The set of NE of Game 3 is:
\begin{align}
E(t) = \{ \underline{a}_{t} \in \mathcal{A}(t) \text{ such that } ||\underline{a}_{t}||_1=1 \text{ or } ||\underline{a}_{t}||_1>2\}.
\end{align}
\end{theorem}
\begin{proof}
The proof of this Theorem can be found in Appendix C.
\end{proof}

These games formulation, thought well defined for a system with perfect and complete information (i.e. all the information in the system are common knowledge by all the players), suffers from many flaws when the information is incomplete. First, by inspection of the NEs of the games, we clearly can see that most of the NEs yield the worst payoff. In all these game formulations, the action profile $\underline{a}_{t} \in \mathcal{A}(t)$ such that $||\underline{a}_{t}||_1>2$ is a NE of the game. The number of this type of action is $\binom{M}{3}$ compared to the $M$ actions profile of our interest (i.e. action profiles in which only one player is transmitting). Secondly, in incomplete information scenario and without a punishment policy, the game can loop infinitely without reducing the Wants set of any player. For these reasons, the overall performance of the games will be very poor in the incomplete information scenario and a more robust definition of the game must be addressed.

\section{Punishment and Bargaining}\label{sec:game2}

In this section, we first build up a punishment policy to prevent multiple collision to occur in the network. We, then, reformulate the problem as a Nash bargaining problem in order to provide a more efficient solution and we provide the expression of the utility of the new version of the completion time, the maximum decoding delay and the sum decoding delay games. We reformulate the games in such a way that only the actions profile of our interest (those that can reduce a Wants set of at least one of the players) are NE of the games. This reformulation as shown at the end of the section, will provide more robust NE and allow the system to operate with more efficiently incomplete information.

\subsection{Punishment and Back-off Function}

In the first definition of the game, after a collision occurs, each of the players that transmitted can re-transmit in the next stage of the game. In order to overcome the scenario in which multiple consecutive collisions occur, we impose a punishment period of $V$ to every player responsible of a collision. In other words, players responsible of a collision will back-off and will not be able to transmit during the next $V$ transmissions. Let $\underline{c}_t = (c_1(t),...,c_M(t))^{\textbf{T}}\in \{0,1\}^M$ be the collision indicator defined as follows:
\begin{align}
c_i(t) =
\begin{cases}
1 \hspace{0.5cm}& \text{if } a_i(t) = 1 \text{ and } ||\underline{a}_t||_1>1 \\
0 \hspace{0.5cm}& \text{otherwise}.
\end{cases}
\end{align}

Let $\textbf{C}$ be the collision history over the last $V$ stage of the game. The mathematical definition of this variable is:
\begin{align}
\begin{tabular}{c|ccc}
$\textbf{C}:$ & $ \mathcal{H}_t$ & $\longrightarrow$ & $\{0,1\}^{M \times V}$  \\
 & $\underline{h}_t$ & $\longmapsto$ & $\textbf{C}(\underline{h}_t) = [\underline{c}_{t-V},...,\underline{c}_{t-1}]$.
\end{tabular}
\end{align}

For notation consistency, the collision indicator for a non positive time index is taken 0 i.e. $\underline{c}_{-t} = \underline{0},~\forall~t \geq 0$. The back-off function indicates at each stage $t$ of the game which players are allowed to transmit. The mathematical definition of this function is:
\begin{align}
\begin{tabular}{c|ccc}
$\underline{\mathcal{B}}:$ & $\{0,1\}^{M \times V}$ & $\longrightarrow$ & $\{0,1\}^{M}$  \\
 & $\textbf{C}(\underline{h}_t)$ & $\longmapsto$ & $ \underline{\mathcal{B}}(\underline{h}_t) = \textbf{C}(\underline{h}_t) \underline{1}$.
\end{tabular}
\end{align}
Let $\mathcal{A}_i^\prime(t)$ be the action space of player $i$ at each stage $t$ of the game defined as follows:
\begin{align}
\mathcal{A}_i^\prime(t) =
\begin{cases}
\mathcal{A}_i(t) \hspace{0.5cm}& \text{if } \underline{\mathcal{B}}_i(\underline{h}_t)=0 \\
\text{\{remain silent\}} \hspace{0.5cm}& \text{otherwise}.
\end{cases}
\end{align}

\subsection{New Game Formulation}

The Nash bargaining solution is a solution concept introduced by Nash initially for the two-player games. Whereas the formulation of the non-cooperative games involves details about the utility functions, the formulation of the cooperative games is generally based on an abstract approach of the problem where only the gain of the collation is presented without stating how this gain is earned. One important idea of the bargaining theory is that solution concepts of cooperative games can be interpreted as non-cooperative games based on individual utility maximization from the knowledge/belief of the players. The construction of non-cooperative bargaining games that sustain multiple cooperative solution concepts as their equilibrium outcomes is mainly based on the idea of having a solution which is fair and more efficient than the one obtained without bargaining or agreement. The bargaining solution is characterized by its Pareto optimality, individual rationality, invariance to positive affine transformation, independence of irrelevant alternatives and symmetry. The following theorem links the bargaining solution to the Pareto-optimal action profile:
\begin{theorem}
For games with unique Pareto-optimal action profile, we can construct a game in which there exists a unique Nash bargaining solution verifying the axioms of the bargaining solution. This solution is the unique Pareto-optimal action profile.
\end{theorem}
\begin{proof}
The proof of this Theorem can be found in \cite{Lasaulce:2011:GTL:2086746}.
\end{proof}

According to Theorem \ref{th2}, the Pareto-optimal action profile of our game is the PONE. Therefore, if the PONE of the game is unique, the bargaining solution is the PONE of the game. We now design the new game to fulfill this condition.

\textbf{Game 4 (New completion time Game):}
\begin{itemize}
\item \textit{Players}: Users in set $\mathcal{M}$
\item \textit{History }: $\underline{h}_t=$ Channel realization $\boldsymbol{\omega}(t)$ and players' action $\underline{a}_t$ at each stage $t \geq 1$.
\item \textit{Strategies}: Contingency plans for selection of a transmission policy at each stage $t \geq 1$ and for any given history $\underline{h}_t$. 
\item \textit{Utilities}: $\mathcal{U}_i^{CT}$ for each player $i$, where at each stage $t \geq 1$ and for any given history $\underline{h}_t$ and action profile $\underline{a}_t$:
\end{itemize}
In this version of the game, we encourage sparsity by the $\ell_1$ regularizer. Moreover, prioritization between players is added when all action profiles will yield a higher maximum excepted completion time than in the previous stage of the game. This additional term is scaled by the number of player to not change the original game. The mathematical definition of the new utility is the following:
\begin{align}
\mathcal{U}_i^{CT}(\underline{a}_t,\underline{h}_t) &= -||\underline{\mathcal{C}}(\underline{a}_t,\underline{h}_t)||_\infty - ||\underline{a}_t||_1 \nonumber \\
& - \cfrac{||\underline{\mathds{D}}(\underline{a}_t,\underline{h}_t)-\underline{\mathds{D}}(\underline{a}_{t-1},\underline{h}_{t-1})||_1}{M}.
\end{align}

\textbf{Game 5 (New maximum decoding delay Game):}
\begin{itemize}
\item \textit{Players}: Users in set $\mathcal{M}$.
\item \textit{Actions Space}: Action in the set $\{\overline{\mathcal{A}^\prime}_i\}_i$.
\item \textit{History }: $\underline{h}_t=$ Channel realization $\boldsymbol{\omega}(t)$ and players' action $\underline{a}_t$ at each stage $t \geq 1$.
\item \textit{Strategies}: Contingency plans for selection of a transmission policy at each stage $t \geq 1$ and for any given history $\underline{h}_t$. 
\item \textit{Utilities}: $\mathcal{U}_i^{MDD}$ for each player $i$, where at each stage $t \geq 1$ and for any given history $\underline{h}_t$ and action profile $\underline{a}_t$.
\end{itemize}
As in the previous game. we encourage sparsity and add a scaled term to take into account the cases when the maximum decoding delay cannot be kept at the same level in the next stage of the game. The new payoff function is the following:
\begin{align}
\mathcal{U}_i^{MDD}(\underline{a}_t,\underline{h}_t) &= -||\underline{\mathds{D}}(\underline{a}_t,\underline{h}_t)||_\infty - ||\underline{a}_t||_1 \nonumber \\
& - \cfrac{||\underline{\mathds{D}}(\underline{a}_t,\underline{h}_t)-\underline{\mathds{D}}(\underline{a}_{t-1},\underline{h}_{t-1})||_1}{M}.
\end{align}

\textbf{Game 6 (New sum decoding delay Game):}
\begin{itemize}
\item \textit{Players}: Users in set $\mathcal{M}$.
\item \textit{Actions Space}: Action in the set $\{\overline{\mathcal{A}^\prime}_i\}_i$.
\item \textit{History }: $\underline{h}_t=$ Channel realization $\boldsymbol{\omega}(t)$ and players' action $\underline{a}_t$ at each stage $t \geq 1$.
\item \textit{Strategies}: Contingency plans for selection of a transmission policy at each stage $t \geq 1$ and for any given history $\underline{h}_t$. 
\item \textit{Utilities}: $\mathcal{U}_i^{SDD}$ for each player $i$, where at each stage $t \geq 1$ and for any given history $\underline{h}_t$ and action profile $\underline{a}_t$.
\end{itemize}
In order to encourage sparsity of the action profile vector $\underline{a}_t$, the $\ell_1$ regularizer is added to the previous definition of the utility function.
\begin{align}
\mathcal{U}_i^{SDD}(\underline{a}_t,\underline{h}_t) = -||\underline{\mathds{D}}(\underline{a}_t,\underline{h}_t)||_1 - ||\underline{a}_t||_1.
\end{align}
As for Game 1, Game 2, and Game 3, Games 4, 5 and 6 are non-cooperative stochastic potential games. In these game formulations, we reduced the set of equilibrium to the one dimensional line (in which only one player is transmitting) containing the Nash bargaining solution of our interest. In the next section, we further proof that the improved definitions of the games have a more stable equilibrium.

\subsection{Equilibrium Stability}

In virtue of Corollary \ref{cor}, in all the games, there exist at least one NE. Moreover, according to Theorem \ref{th2}, the maximum of the utility function is the PONE of the game. However, the existence of the NE or the PONE is not sufficiency. We first define the price of anarchy, introduced in \cite{koutsoupias1999worst} to be able to characterize the equilibrium in a game. 
\begin{definition}
The price-of-anarchy (PoA), at stage $t$, is the worst-case efficiency of a Nash Equilibrium among all possible strategies. In other words, the PoA is defined as:
\begin{align}  
PoA(t) = \cfrac{{\text{max}_{s \in S(t)}} W(s)}{\text{min}_{s \in E(t)} W(s)},
\end{align}
where:
\begin{itemize}
\item $S(t)$ is the set of all possible strategies at stage $t$,
\item $E(t)$ is the set of all NE at stage $t$,
\item $W:S\longrightarrow\mathds{R}$ is a fairness function .
\end{itemize}
\end{definition}

The PoA is a concept that measures how the efficiency of a game degrades due to selfish behavior of its players. Since we have in our game the Pareto Equilibrium is a Nash equilibrium and the utility function is the same for all the players (the corresponding coordination game is the game itself), then the previous definition reduces to the following:
\begin{align}  
PoA(t) = \cfrac{{\text{max}_{s \in E(t)}} \phi(s)}{\text{min}_{s \in E(t)} \phi(s)}.
\end{align}

The utility function is always negative in our context. Therefore, the PoA is a well defined quantity. In this paper, since the utility function is strictly negative, we will compute the PoA using the cost function $\phi^\prime=-\phi$. The PoA can be expressed in terms of the cost function as follows:
\begin{align}  
PoA(t) = \cfrac{{\text{min}_{s \in E(t)}} \phi^\prime(s)}{\text{max}_{s \in E(t)} \phi^\prime(s)}.
\end{align}

Note that, since the Pareto-Optimal action profile is a NE of the game, the price of stability \cite{Nisan:2007:AGT:1296179} in our game is equal to $1$ and can not be used to characterize efficiency. The following theorem gives the \emph{PoA} of the Game 1 and Game 4:
\begin{theorem}
The \emph{PoA} of Game 1 can be expressed as follows:
\begin{align}
&PoA(t) = \nonumber \\
&\begin{cases}
1 \hspace{0.5cm} &\text{if } Z(t) = \varnothing \\
1 - \cfrac{Y_0(t) - \underset{j \in Z(t)}{min }(Y_j(t))}{\phi^\prime (\underline{a}_{t-1},\underline{h}_{t-1})+Y_0(t)} \hspace{0.5cm} &\text{otherwise },
\end{cases}
\end{align}
where $\phi^\prime=-\mathcal{U}_i^{CT},~\forall~i\in \mathcal{M}$ is the cost of Game 1.
The \emph{PoA} of Game 4 can be expressed as follows:
\begin{align}
&PoA^\prime(t) = \\
&\begin{cases}
\cfrac{\underset{||\underline{a}_t||_1=a_i(t)=1}{\text{min}}\phi^\prime (\underline{a}_{t-1},\underline{h}_{t-1})+\cfrac{||\underline{\mathcal{D}}_{\omega_i,\kappa^i}||_1}{M} +1 + Y_i(t) }{\underset{||\underline{a}_t||_1=a_i(t)=1}{\text{max}}\phi^\prime (\underline{a}_{t-1},\underline{h}_{t-1})+\cfrac{||\underline{\mathcal{D}}_{\omega_i,\kappa^i}||_1}{M} +1 + Y_i(t)} \\
\hspace{6cm} \text{ if } Z(t) \neq \varnothing \\
\cfrac{\underset{||\underline{a}_t||_1=a_i(t)=1}{\text{min}}\phi^\prime (\underline{a}_{t-1},\underline{h}_{t-1})+\cfrac{||\underline{\mathcal{D}}_{\omega_i,\kappa^i}||_1}{M}  +1 + Y_0(t) }{\underset{||\underline{a}_t||_1=a_i(t)=1}{\text{max}}\phi^\prime (\underline{a}_{t-1},\underline{h}_{t-1})+\cfrac{||\underline{\mathcal{D}}_{\omega_i,\kappa^i}||_1}{M} +1 + Y_0(t)} \\
\hspace{6cm} \text{ otherwise }.
\end{cases} \nonumber
\end{align}
where $\phi^\prime=-\mathcal{U}_i^{CT},~\forall~i\in \mathcal{M}$ is the cost of Game 4.
\end{theorem}
\begin{proof}
The proof of this Theorem can be found in Appendix A for Game 1 and Appendix D for Game 4.
\end{proof}
For stage $t$ of the game with $Z(t) = \varnothing$, the expected completion time will increase anyway and thus the \emph{PoA} do not have a practical signification in this case. In stage $t$ where $Z(t) \neq \varnothing$, we have:
\begin{align}
PoA^\prime(t) > PoA(t).
\end{align}
This new definition of the game offers a more efficient equilibrium.

The following theorem gives the \emph{PoA} of the Game 2 and Game 5:
\begin{theorem}
The \emph{PoA} of Game 2 can be expressed as follows:
\begin{align}
PoA(t) = 
\begin{cases}
1 \hspace{0.5cm} &\text{if } Z(t) = \varnothing \\
1 - \cfrac{1}{\phi^\prime (\underline{a}_{t-1},\underline{h}_{t-1})+1} \hspace{0.5cm} &\text{otherwise }.
\end{cases}
\end{align}
where $\phi^\prime=-\mathcal{U}_i^{MDD},~\forall~i\in \mathcal{M}$ is the cost of Game 2. The \emph{PoA} of Game 5 can be expressed as follows:
\begin{align}
&PoA^\prime(t) = \\
&\begin{cases}
\cfrac{\underset{||\underline{a}_t||_1=a_i(t)=1}{\text{min}}\phi^\prime (\underline{a}_{t-1},\underline{h}_{t-1})+\cfrac{||\underline{\mathcal{D}}_{\omega_i,\kappa^i}||_1}{M} +2 }{\underset{||\underline{a}_t||_1=a_i(t)=1}{\text{max}}\phi^\prime (\underline{a}_{t-1},\underline{h}_{t-1})+\cfrac{||\underline{\mathcal{D}}_{\omega_i,\kappa^i}||_1}{M} +2} \\
\hspace{6cm} \text{ if } Z(t) = \varnothing \\
\cfrac{\underset{||\underline{a}_t||_1=a_i(t)=1}{\text{min}}\phi^\prime (\underline{a}_{t-1},\underline{h}_{t-1})+\cfrac{||\underline{\mathcal{D}}_{\omega_i,\kappa^i}||_1}{M} +1 }{\underset{||\underline{a}_t||_1=a_i(t)=1}{\text{max}}\phi^\prime (\underline{a}_{t-1},\underline{h}_{t-1})+\cfrac{||\underline{\mathcal{D}}_{\omega_i,\kappa^i}||_1}{M} +1} \\
\hspace{6cm} \text{ otherwise }.
\end{cases} \nonumber
\end{align}
where $\phi^\prime=-\mathcal{U}_i^{MDD},~\forall~i\in \mathcal{M}$ is the cost of Game 5.
\end{theorem}
\begin{proof}
The proof of this Theorem can be found in Appendix B for Game 2 and Appendix E for Game 5.
\end{proof}
For stage $t$ of the game with $Z(t) = \varnothing$, the maximum delay will increase anyway and thus the \emph{PoA} do not have a practical signification in this case. In stage $t$ where $Z(t) \neq \varnothing$, we have:
\begin{align}
PoA^\prime(t) > PoA(t).
\end{align}
This new definition of the game offers a more efficient equilibrium.

The following theorem gives the \emph{PoA} of the Game 3 and Game 6:
\begin{theorem}
The \emph{PoA} of Game 3 can be expressed as follows:
\begin{align}
PoA(t) = \cfrac{\underset{||\underline{a}_t||_1=a_i(t)=1}{\text{min}}\phi^\prime (\underline{a}_{t-1},\underline{h}_{t-1})+||\underline{\mathcal{D}}_{\omega_i,\kappa^i}||_1}{\phi^\prime (\underline{a}_{t-1},\underline{h}_{t-1})+||\underline{M}^w(t)||_1}.
\end{align}
where $\phi^\prime=-\mathcal{U}_i^{SDD},~\forall~i\in \mathcal{M}$ is the cost of Game 3.
The \emph{PoA} of Game 6 can be expressed as follows:
\begin{align}
PoA^\prime(t) = \cfrac{\underset{||\underline{a}_t||_1=a_i(t)=1}{\text{min}}\phi^\prime (\underline{a}_{t-1},\underline{h}_{t-1})+||\underline{\mathcal{D}}_{\omega_i,\kappa^i}||_1 +1 }{\underset{||\underline{a}_t||_1=a_i(t)=1}{\text{max}}\phi^\prime (\underline{a}_{t-1},\underline{h}_{t-1})+||\underline{\mathcal{D}}_{\omega_i,\kappa^i}||_1 +1}.
\end{align}
where $\phi^\prime=-\mathcal{U}_i^{SDD},~\forall~i\in \mathcal{M}$ is the cost of Game 6.
\end{theorem}
\begin{proof}
The proof of this Theorem can be found in Appendix C for Game 3 and Appendix F for Game 6.
\end{proof}

Clearly, we have:
\begin{align}
\underset{||\underline{a}_t||_1=a_i(t)=1}{\text{min}}||\underline{\mathcal{D}}_{\omega_i,\kappa^i}||_1  \leq ||\underline{M}^w(t)||_1 -1.
\end{align}
Therefore, we obtain:
\begin{align}
PoA^\prime(t) \geq PoA(t).
\end{align}
This new definition of the game offers a more efficient equilibrium. 

After investigation of the Price-Of-Anarchy of the two versions of the games, the new formulation of the games offers a more efficient equilibrium. Therefore, for any learning algorithm and a long running period, the second formulation will performs better in terms of delay than the first one. 

\section{Distributed Learning Algorithms}\label{sec:algo}

In this section, we introduce the learning algorithm used to simulate the system in two settings. First, we assume complete and perfect information. The former assumption means that the actions available to the players and the utility functions are common knowledge i.e. every player knows the data of the game, every player knows that the other players know the data of the game, every player knows that every player knows that every players knows the data of the game, and so on, ad infinitum. The latter assumption means all the players know the history of the game perfectly. Note that these assumptions do not add extra constraints to the problem but are rather intrinsic to it when the system operate with full feedback. We use the best response algorithm to simulate such a system.

In the second part, we propose a low complexity system in which the players do not compute the packet combination that the other players can make. In term of game theory terms, this assumption means that the system will operate with incomplete information. Although all the information is available at every node, the action set of other player is not common knowledge between players and each player know only its action set. This application is motivated by the wireless sensor network. In this type of network, the nodes do not have the necessary power to do the computation of all the packet combination that other node can make and thus can sacrifice optimality to save energy by making less efficient decision. To simulate this system, we will make use of the reinforcement learning algorithm.

\subsection{Best Response Algorithm}

In the original formulation of the best response algorithm, by Cournot \cite{cournot1838recherches}, players choose their actions sequentially. At each time slot, a player selects the action that is the best response to the action chosen by the other players in the previous time slot. Since the state of the game is not known to players, the utility function will be replaced by the expected utility function. This can be done by replacing the actual state $\underline{\omega}_i(t)$ by its expected value $\underline{P}_i(t)$ in \eref{qwer}. The following theorem characterize the outcome of the best-response algorithm for our games:

\begin{theorem}
For our games with perfect and complete information, the best-response algorithm will make the system operate in the PONE of the game.
\end{theorem}
\begin{proof}
To proof this theorem, we fist introduce the following theorem:
\begin{theorem}
Let $\mathcal{G}$ be a best-reply potential game with $\mathcal{V}$ a best response potential. If the action $\underline{a}_t^*$ maximizes $\mathcal{V}$, then $\underline{a}_t^*$ is a NE.
\label{thhhh}
\end{theorem}
\begin{proof}
The proof of this theorem can be found in \cite{tembine2012distributed}.
\end{proof}
From our previous analysis, the NE of the games are located on the one dimensional line in which only one player is transmitting. Let $\underline{a}_t^*$ be the PONE of the game such that $||\underline{a}_t^*||_1=a_i(t)=1$. Assume that the outcome of the best-response algorithm is the action profile $\underline{a}_t^\prime \neq \underline{a}_t^*$. In virtue of Theorem \ref{thhhh} and our previous analysis, the action profile will have $||\underline{a}_t^\prime||_1=a_j(t)=1, j\neq i$. For simplicity, assume that players take action sequentially in order. Assume first that $j<i$ i.e. player $j$ will take action before player $i$. Therefore, player $j$ did not take its best-response game because he can insure better payoff by choosing not to transmit. Hence, we obtain $j\geq i$. Assume now $j>i$. Since player $i$ is taken its best action then it will transmit. According to theorem \ref{thhhh} player $j$ is unable to transmit otherwise the outcome will not be a NE of the game. Therefore we obtain $i=j$. In other words, the only outcome of the best-response algorithm is the PONE of the game.
\end{proof}

In this setting, every node compute the packet combination that the other nodes can make, allowing it to completely compute the payoff function. Therefore, this setting is equivalent to the centralized setting since all the computation is done at each and every node. Note that with complete and perfect information, as stated previously, the system will work on the PONE of the game. Therefore, in such a system, no collisions will occurs. In that case, the cooperative decoding delay will be equivalent to the conventional decoding delay and the system will not suffer from any additional delay due to the cooperation between players.

\subsection{Reinforcement Learning Algorithm}

Reinforcement learning \cite{35451549} was originally studied in the context of a single player game. It was formalized in \cite{158125469} to study the behavior of animals. The idea behind this algorithm is that the players interact with their environment and depending on the outcome of the past actions, they decide to choose or avoid certain actions. Actions, that led to a high (satisfactory) payoff by the past, are repeated if the same situation occurs. Actions that were not satisfactory tend to be avoided. In a game theory context, reinforcement learning is implemented by associating a probability distribution over the set of possible actions for all players and update these distributions according to the perceived payoff at each stage of the game. In this paper, we will use the reinforcement algorithm introduced by Bush and Mosteller \cite{55651} for stochastic games with two-players two-actions. In this algorithm, players select stochastically the actions to take according to their probabilities. Formally, the probability $x_j$ that player $j$ take the action $s_j$ at stage $t+1$ is calculated as follows:
\begin{align}
x_{j,t+1} (s_j) = 
\begin{cases}
x_{j,t} (s_j) + \lambda_j \overline{s}_{j,t}(1-x_{j,t} (s_j)) \text{ if } &\overline{s}_{j,t}\geq 0 \\
x_{j,t} (s_j) + \lambda_j \overline{s}_{j,t}x_{j,t} (s_j) \text{ if } &\overline{s}_{j,t} < 0,
\end{cases}
\end{align}
where $\lambda_j$ is the learning rate of player $j$ ($0<\lambda_j<1$), and $\overline{s}_{j,t}$ is the stimulus of action $a_i$ ($-1<\overline{s}_{j,t}<1$) defined as follows:
\begin{align}
\overline{s}_{j,t} = \cfrac{\phi_j(a_i,t)-M_j}{\underset{a \in \mathcal{A}}{max}|\phi_j(a,t)-M_j|},
\end{align}
with $\phi_j(a,t)$ is the payoff perceived by player $j$ at time $t$ after taking the action $a$ and $M_j$ is a satisfactory level for player $j$.

For the potential game class, the reinforcement algorithm converge to a NE of the game with probability $1$ \cite{tembine2012distributed} for any learning rate and satisfactory threshold. For simplicity purposes, we will not study, in this paper, the optimal choice of the learning rate or the satisfactory threshold. The learning rate of a player $i$ will be chosen to be the ratio of its Has set $|\mathcal{H}_i|$ by the total number of packets $N$ and the satisfactory criterion will be $-|M_w|_1+\epsilon$ for some $\epsilon > 0$. In other words, the player is satisfied if at least one player is targeted (any action profile that will not yield the worst payoff).

This algorithm will be fluctuating between a stabilization and a stable states. In the stabilization state, players with the biggest Has set will be more likely to transmit. If a collision occurs due to multiple simultaneous transmissions, the responsible players will back-off. This process is repeated until only one player transmits. When this occurs, the system reaches a stable state in which this player will continue to transmit in the following stages until he become enable to satisfy any other player. When this occur, the payoff will not be satisfactory any longer and the system switches to a new stabilization phase. Note that in the stable phase only one player will be computing the optimal packet combination that it can make and all the other player will save energy since they will not transmit in the next stage.

\section{Extension To The Limited Feedback Scenario}\label{sec:ext}

In this section, we extend our previous study to the limited feedback scenario. Without loss of generality, we will focus our attention on the lossy feedback scenario in this section. In this configuration, the feedback can be erased at some players and thus can create uncertainties in the feedback matrix and the overall system. Other limited feedback scenarios can be considered and the same analysis holds. In game theory terms, the lossy feedback scenarios means that neither the information of the game nor the action sets are common knowledge between players. To model this scenario, we use a local system model \cite{letterarxiv} at each player to take account of the uncertainties. At each stage of the game, for every player $j$, three sets of packets are attributed to each player $i\neq j$:
\begin{itemize}
\item The Has set (denoted by $\mathcal{H}_i^j$): the set of packets received and acknowledged by player $i$ to player $j$.
\item The Wants set (denoted by $\mathcal{W}_i^j$): the set of the lost packets or those which are not acknowledged by user $i$ to player $j$.
\item The Uncertain set (denoted by $\mathcal{U}_i^j$) is defined as the set of packets of player $i$ whose state is uncertain at player $j$. We have $\mathcal{U}_i^j \subseteq \mathcal{W}_i^j$.
\end{itemize}

Each player $k$ stocks these information in a \emph{local feedback matrix} $\mathbf{F}^k = [f_{ij}^k],~ \forall~ i \in \mathcal{M},~ \forall~j,k \in\mathcal{N}$ as follows:
\begin{align}
f_{ij}^k =
\begin{cases}
0 \hspace{0.9 cm}& j \in \mathcal{H}_i^k \\
1 \hspace{0.9 cm}& j \in \mathcal{W}_i^k \setminus \mathcal{U}_i^k \\
x \hspace{0.9 cm}& j \in \mathcal{U}_i^k.
\end{cases}
\end{align} 

Using its local feedback matrix, each player $k$ can compute the optimal packet combination that it can encode. The expressions of these optimal packet combination can be found in \cite{jsaccompletion} for the completion time and the maximum decoding delay and in \cite{refjournal} for the sum decoding delay. In this paper, we will assume reciprocal channel i.e. the erasure probability in the forward link (transmission link) is equal to the erasure on the backward link (feedback link). Note that the assumption is made without loss of generality since the expressions of the optimal packet combination are available for general channels.

To simulate the system, we will use a modified version of the reinforcement algorithm introduced in \sref{sec:algo}. Note that each player knows its Has set and thus the learning rate is a well defined quantity. However, the stimulus is not since the players do no have access to the real payoff their action generated.

The reinforcement learning \cite{556dv51} consists in updating the probability distribution of action $s_j$ to be taken in time $t+1$ as follows:

\begin{align}
x_{j,t+1} (s_j) = x_{j,t}(s_j) + \lambda_j \tilde{s}_{j,t}(\mathds{1}_{\{a_{j,t=s_j}\}}-x_{j,t}(s_j)) ,
\end{align}
where $\tilde{s}_{j,t}$ is an estimate of the payoff at stage $t$ of the game and $\mathds{1}$ is the step function. In our context, we will estimate the payoff by its expected value computed when computing the packet combination of each of the delay aspects.

\section{Simulation Results}\label{sec:simul}

In this section, we present the simulation results comparing the delay encountered by players when applying the PMP system against the distributed cooperative data exchange system in different scenarios. We first compare our optimal games (complete and perfect information) of the completion time, the maximum decoding delay and the sum decoding delay against the centralized PMP solution proposed in \cite{confarxiv} for the completion time, in \cite{vtc} for the maximum decoding delay and in \cite{ref2} for the sum decoding delay. We, then, do the same comparison for our low complexity scenarios and finally for the lossy feedback scenario. We also compare the delay experienced by players against the player-player packet erasure probability relatively to the base station-player packet erasure probability since the short range communications are more reliable than the base station-player communications \cite{6620795,6404743}.

In these simulations, the delay is computed over a large number of iterations and the average value is presented. We assume that the packet erasure probability remains constant during a delivery period and change from iteration to iteration while keeping its mean, $P$ of the player-player and $Q$ for the base BS-players, constant. We also assume that each player have perfect knowledge of the packet erasure probabilities linking him to the other player (i.e. the estimation of this probability is perfect).

\begin{figure}[t]
\centering
  \includegraphics[width=1\linewidth]{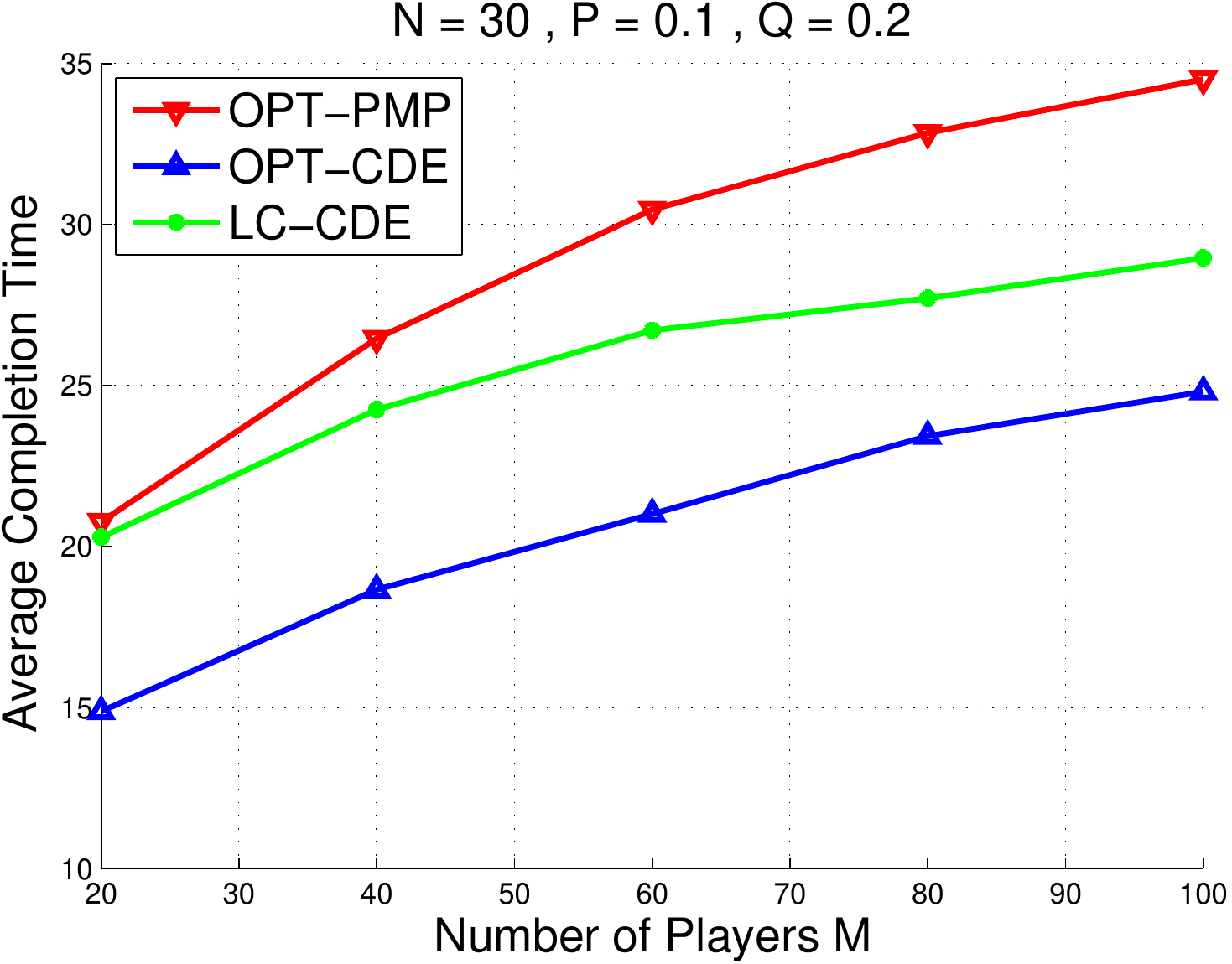}\\
  \caption{Mean completion time versus number of players $M$.}\label{fig:MCT}
\end{figure}

\begin{figure}[t]
\centering
  \includegraphics[width=1\linewidth]{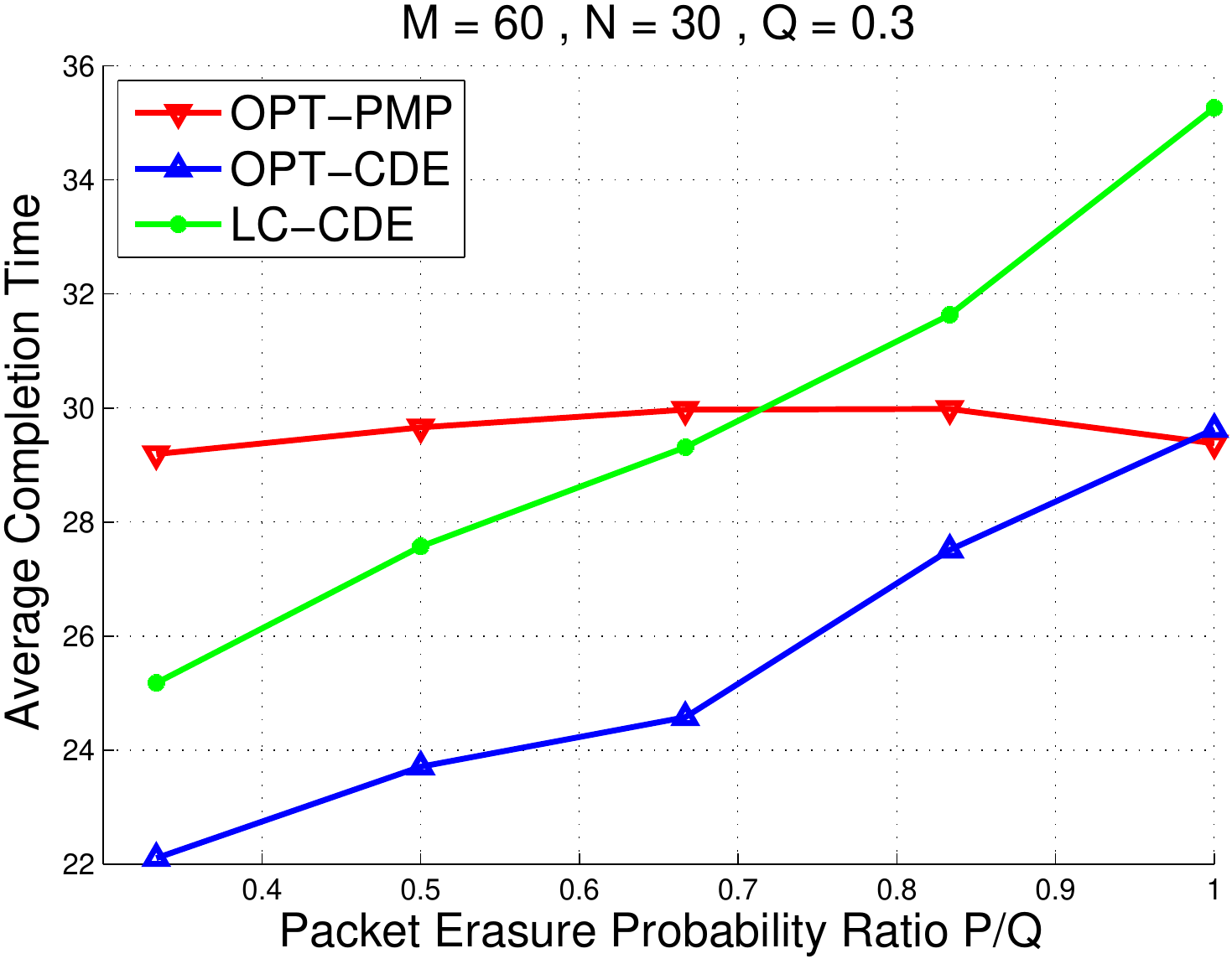}\\
  \caption{Mean completion time versus players erasure $P$.}\label{fig:PQCT}
\end{figure}

\begin{figure}[t]
\centering
  \includegraphics[width=1\linewidth]{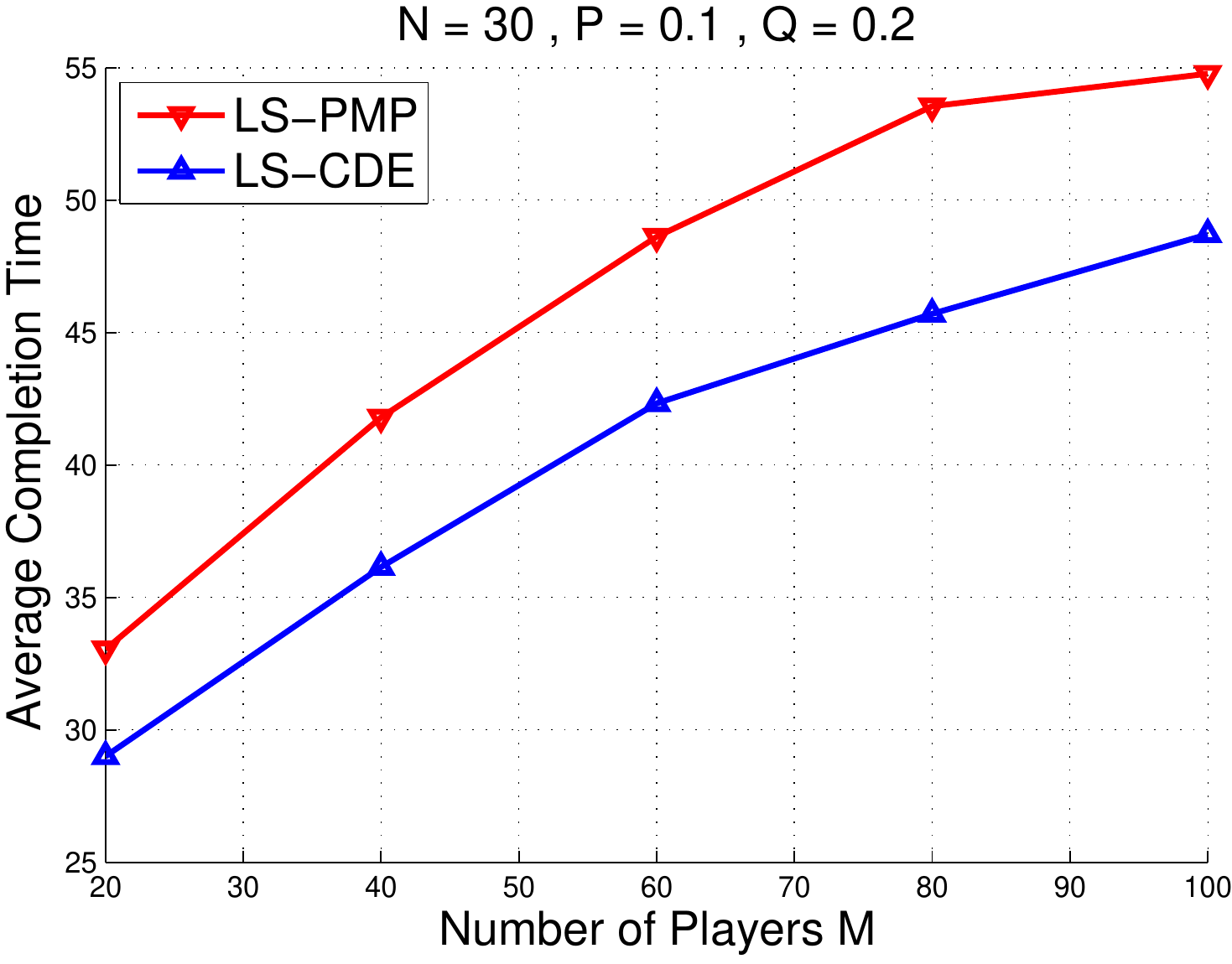}\\
  \caption{Mean completion time in lossy feedback versus number of players $M$.}\label{fig:MCTL}
\end{figure}

\begin{figure}[t]
\centering
  \includegraphics[width=1\linewidth]{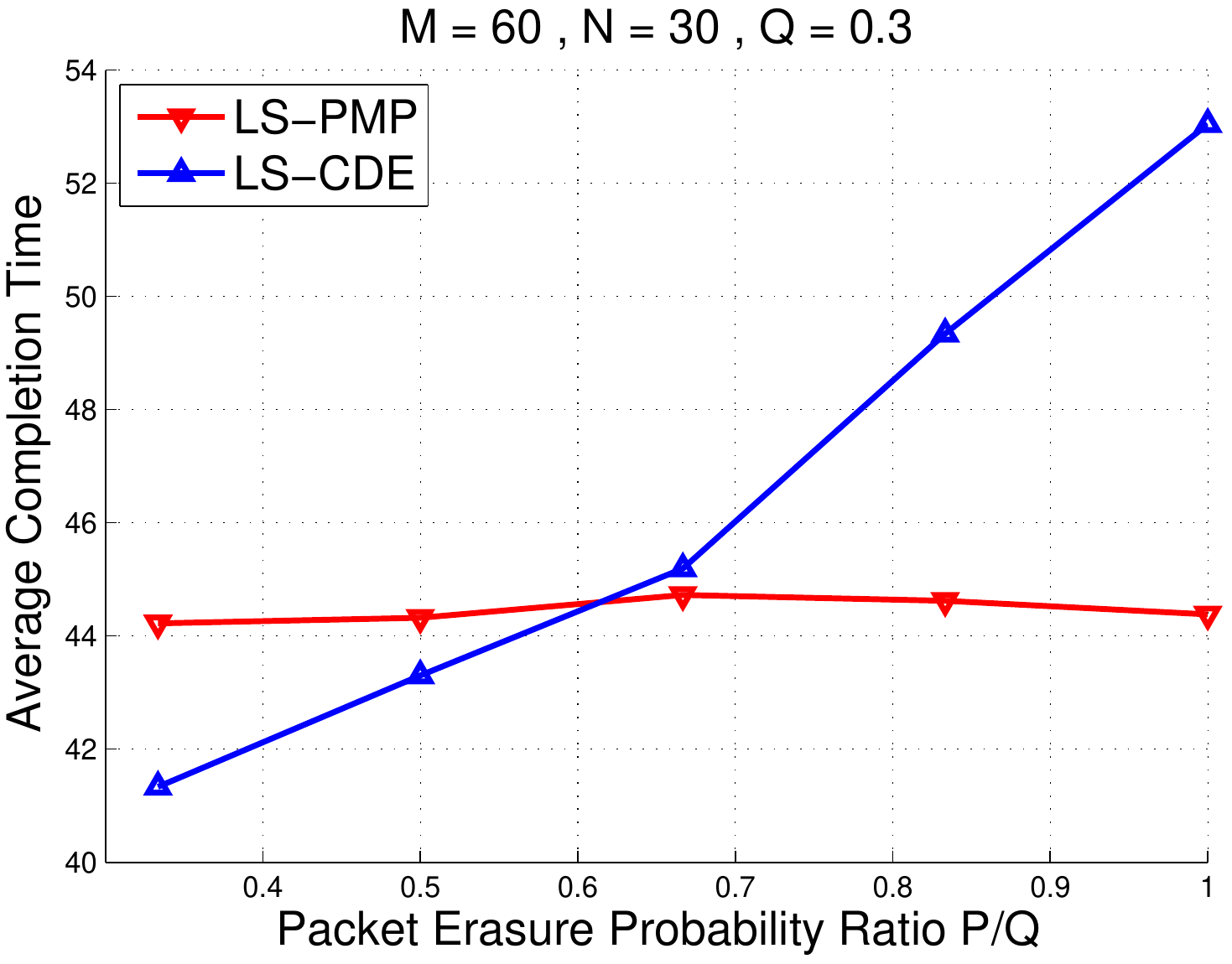}\\
  \caption{Mean completion time in lossy feedback versus players erasure $P$.}\label{fig:PQCTL}
\end{figure}

\begin{figure}[t]
\centering
  \includegraphics[width=1\linewidth]{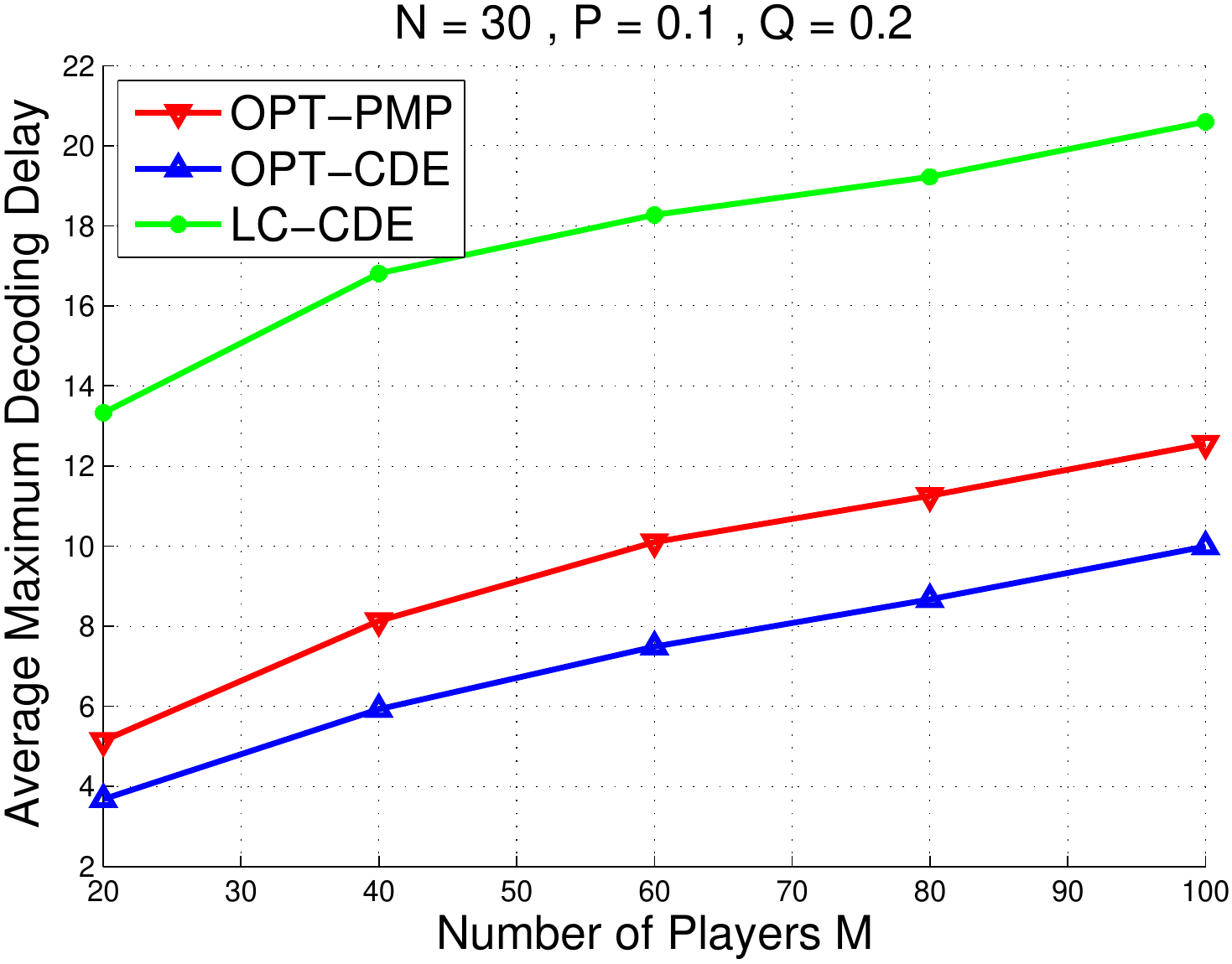}\\
  \caption{Mean maximum delay versus number of players $M$.}\label{fig:MMAX}
\end{figure}

\begin{figure}[t]
\centering
  \includegraphics[width=1\linewidth]{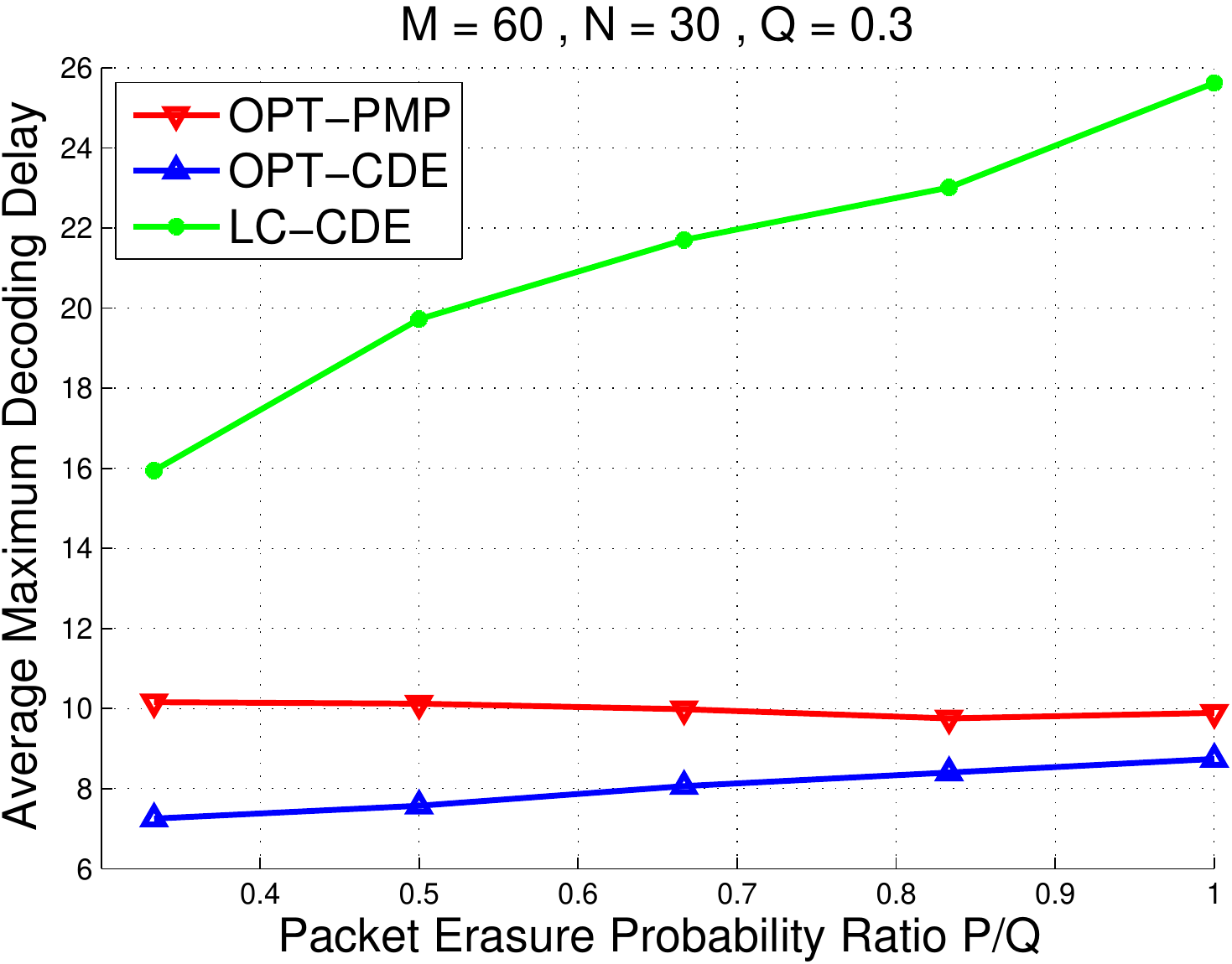}\\
  \caption{Mean maximum delay versus players erasure $P$.}\label{fig:PQMAX}
\end{figure}

\begin{figure}[t]
\centering
  \includegraphics[width=1\linewidth]{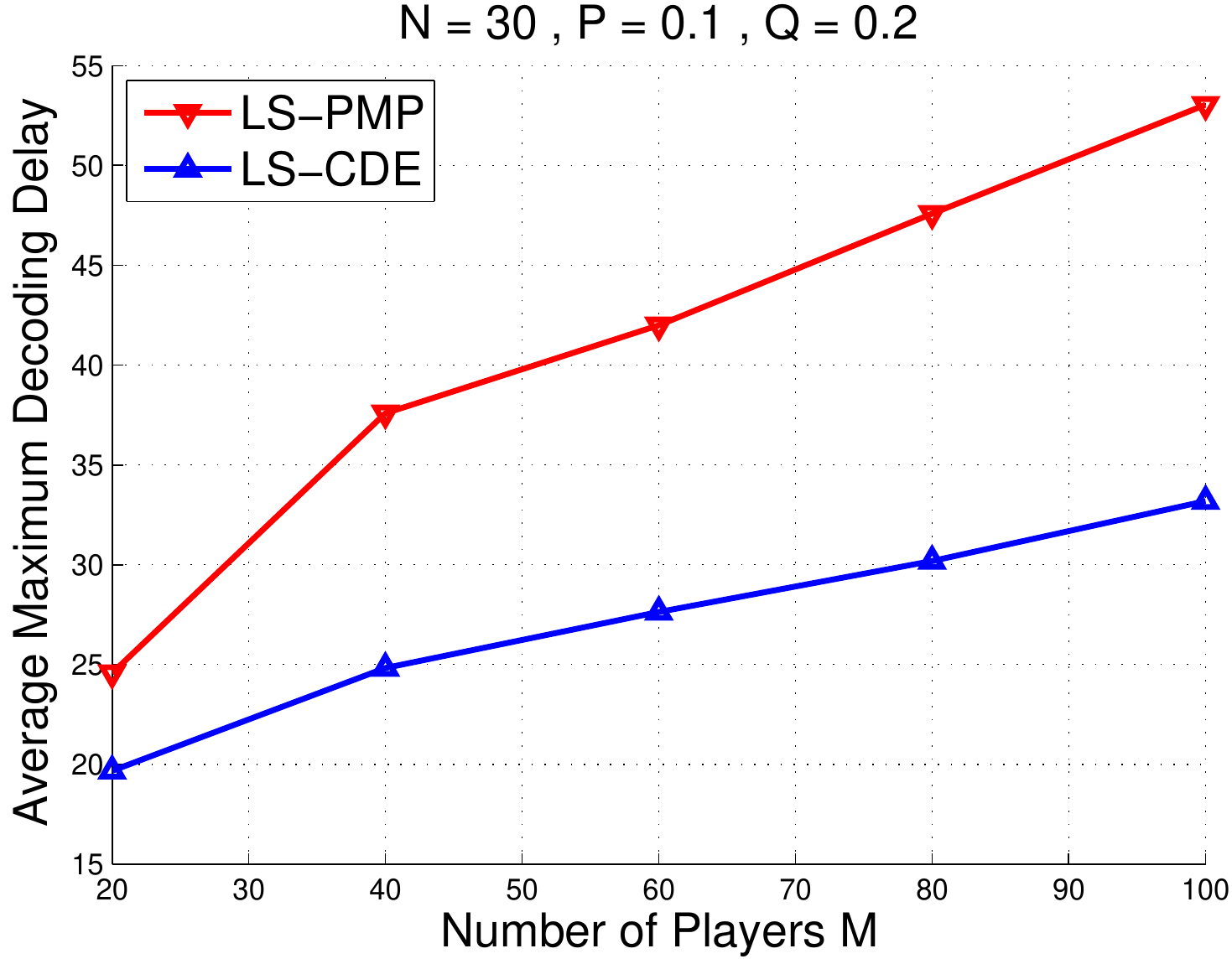}\\
  \caption{Mean maximum delay in lossy feedback versus number of players $M$.}\label{fig:MMAXL}
\end{figure}

\begin{figure}[t]
\centering
  \includegraphics[width=1\linewidth]{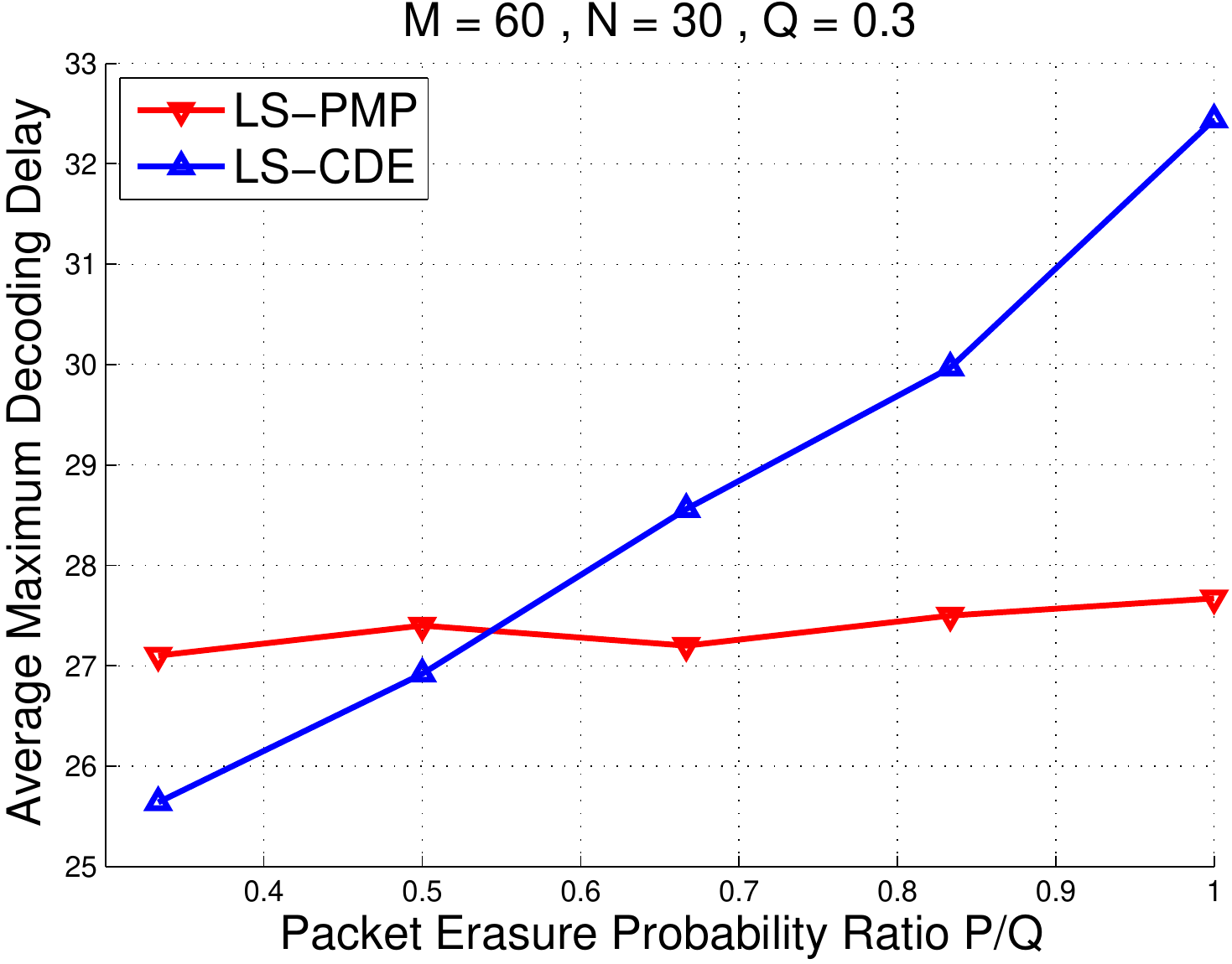}\\
  \caption{Mean maximum delay in lossy feedback versus players erasure $P$.}\label{fig:PQMAXL}
\end{figure}

\begin{figure}[t]
\centering
  \includegraphics[width=1\linewidth]{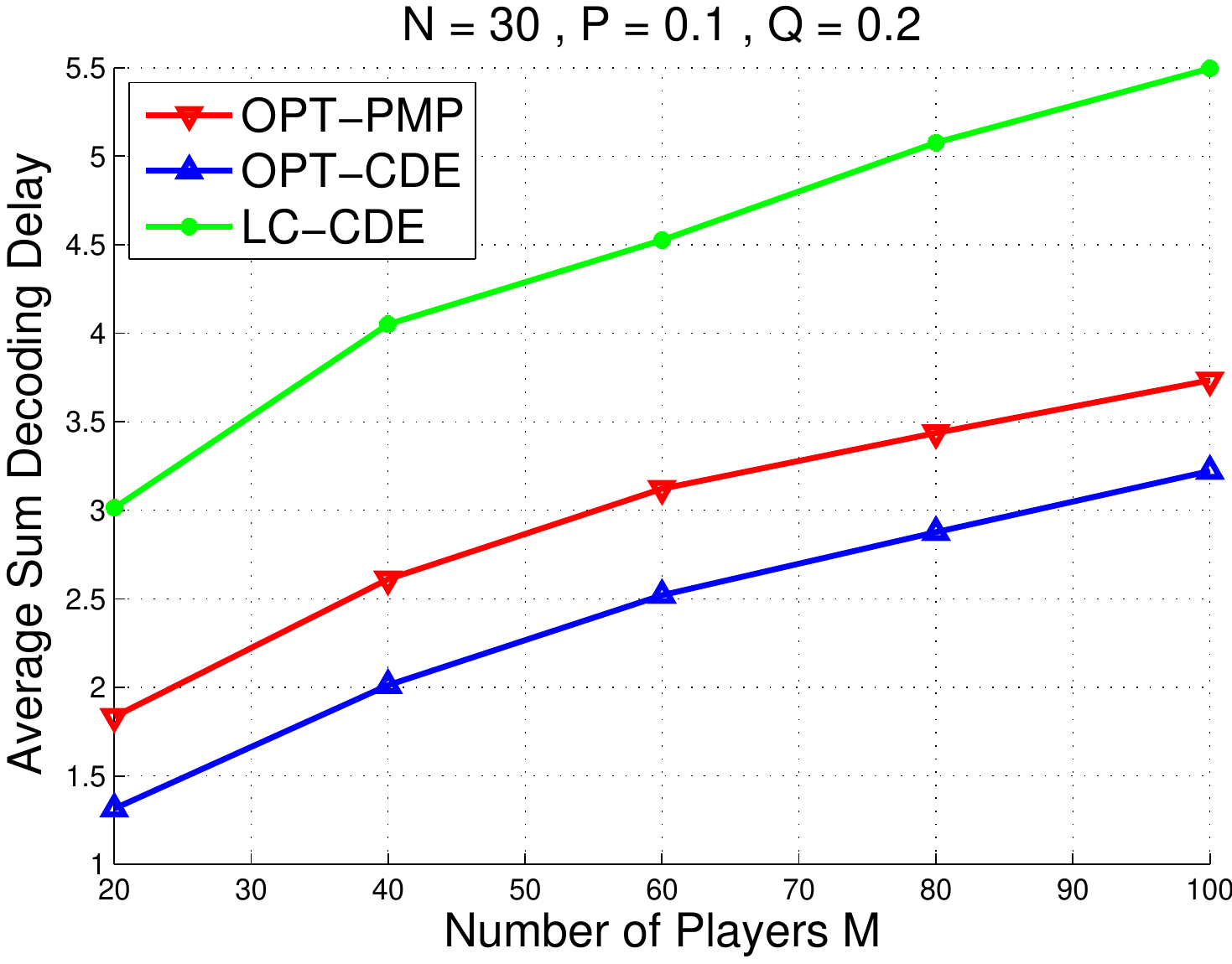}\\
  \caption{Mean sum decoding delay versus number of players $M$.}\label{fig:MSUM}
\end{figure}

\begin{figure}[t]
\centering
  \includegraphics[width=1\linewidth]{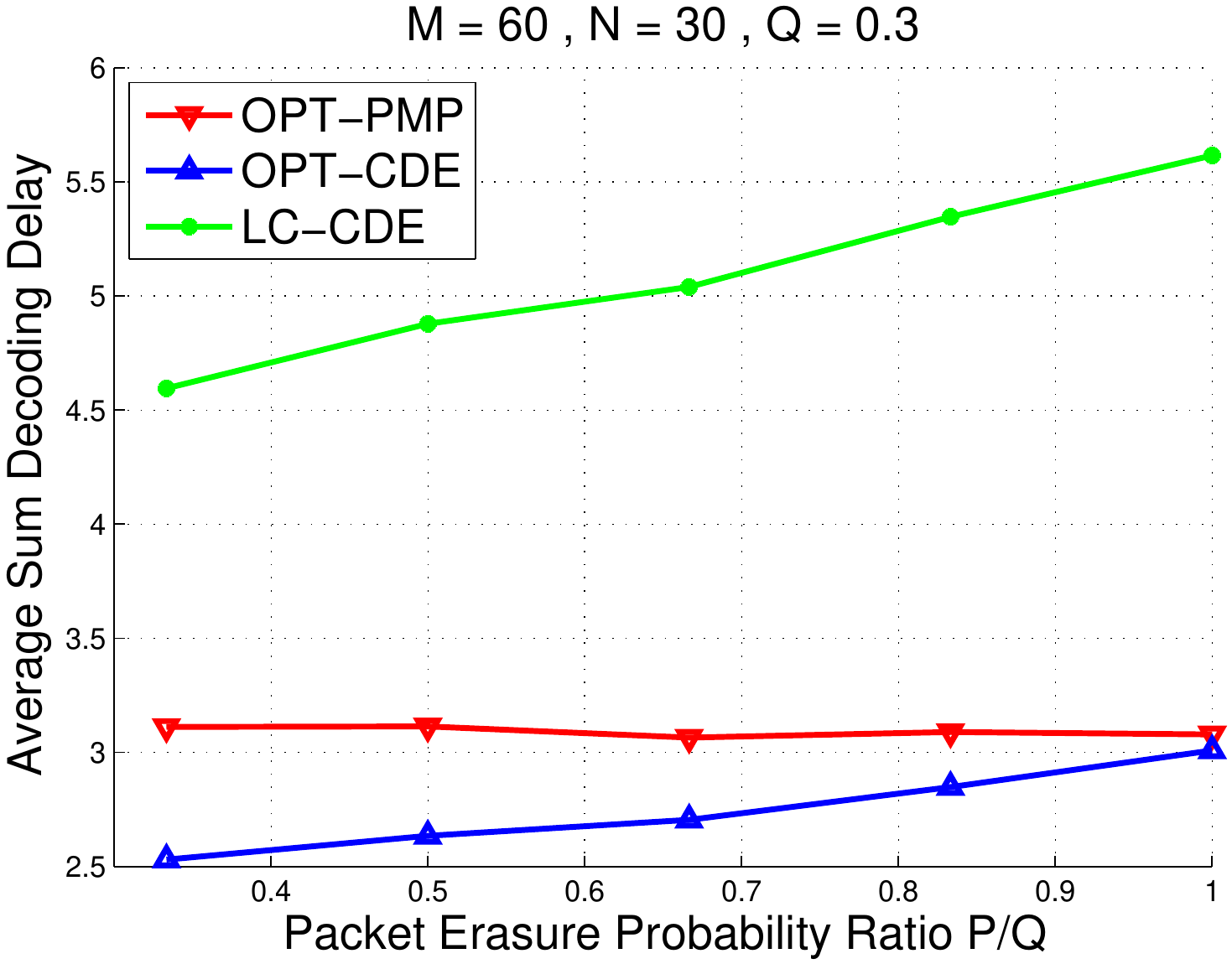}\\
  \caption{Mean sum decoding delay versus players erasure $P$.}\label{fig:PQSUM}
\end{figure}

\begin{figure}[t]
\centering
  \includegraphics[width=1\linewidth]{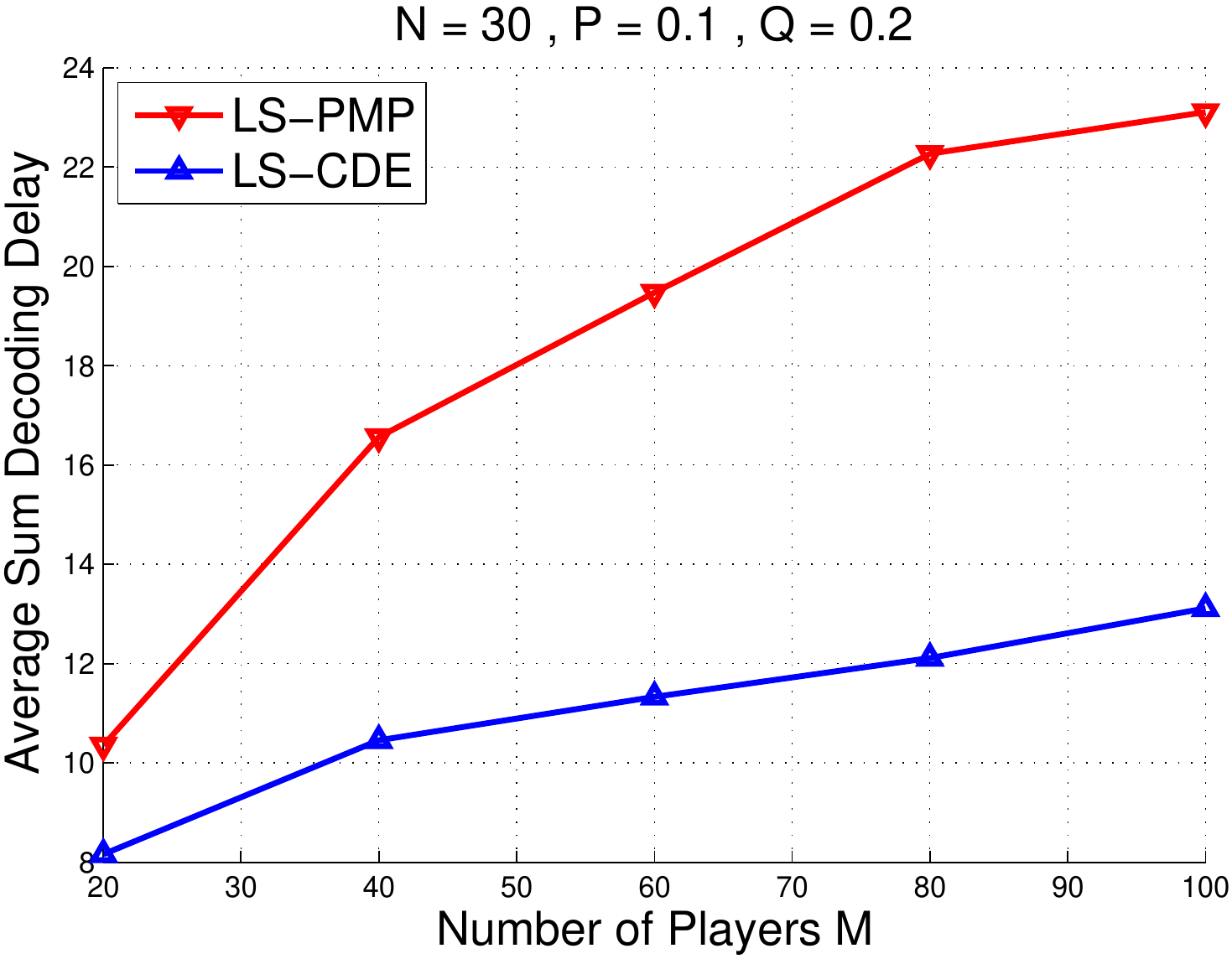}\\
  \caption{Mean sum decoding delay in lossy feedback versus number of players $M$.}\label{fig:MSUML}
\end{figure}

\begin{figure}[t]
\centering
  \includegraphics[width=1\linewidth]{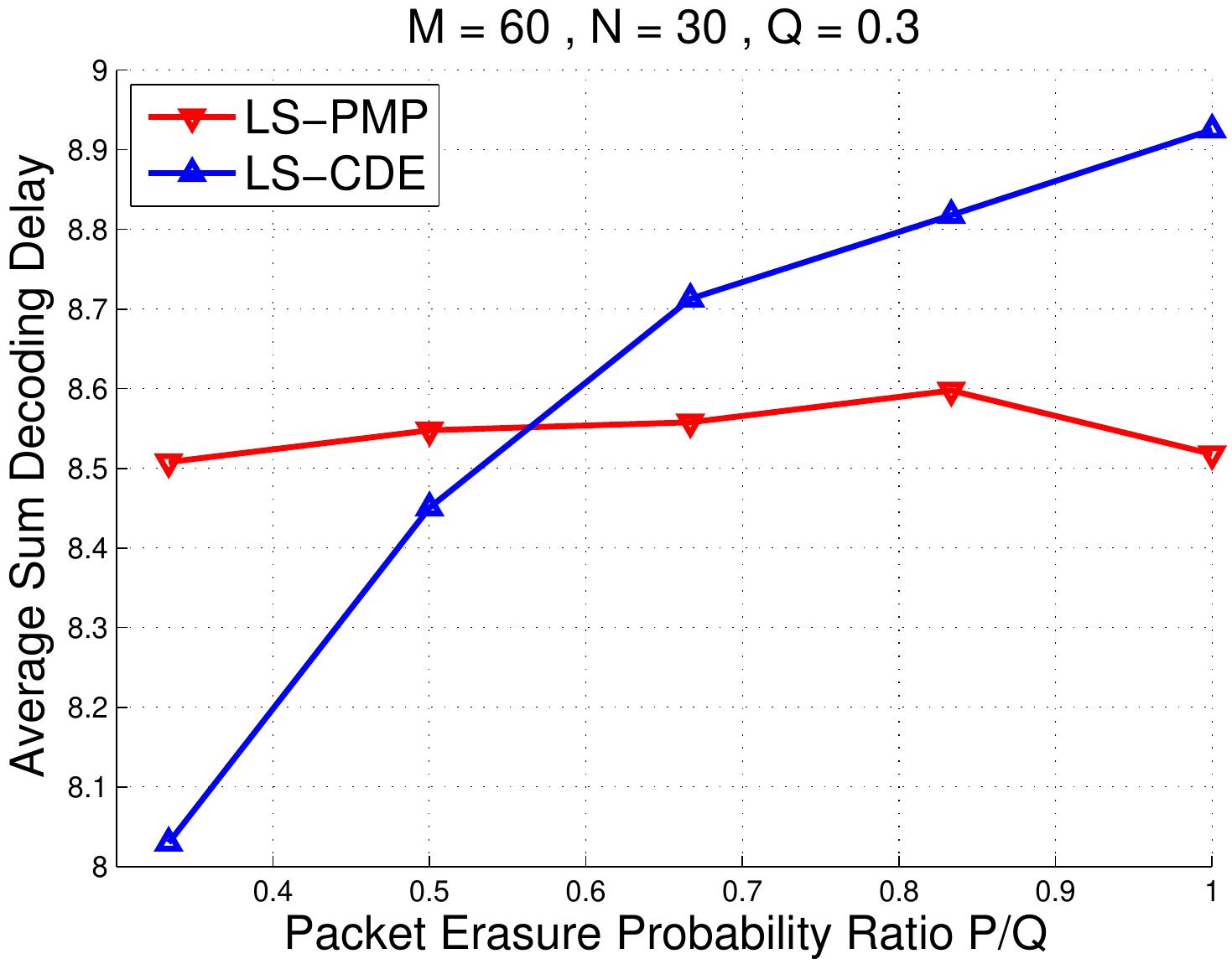}\\
  \caption{Mean sum decoding delay in lossy feedback versus players erasure $P$.}\label{fig:PQSUML}
\end{figure}

\fref{fig:MCT} depicts the comparison of the average completion time achieved by the PMP scheme (denoted by OPT-PMP), our perfect CDE scheme (denoted by OPT-CDE) and our low complexity scenario (denoted by LC-CDE) in perfect feedback scenario against the number of players $M$ for $N=30$, $Q=0.2$, and $P=0.1$. \fref{fig:PQCT} presents the mean completion time against the ratio of the player-player and BS-player erasure probability $P/Q$ for $M=60$, $N=30$, and $Q=0.3$. 

\fref{fig:MCTL} depicts the average completion time achieved by the PMP scheme (denoted by LS-PMP) and our CDE scheme (denoted by LS-CDE) in lossy feedback scenario against the number of players $M$ for $N=30$, $Q=0.2$, and $P=0.1$. \fref{fig:PQCTL} presents the mean completion time in lossy feedback against the ratio of the player-player and BS-player erasure probability $P/Q$ for $M=60$, $N=30$, and $Q=0.3$. 

\fref{fig:MMAX} and \fref{fig:MSUM} illustrate the same comparison in perfect feedback scenario for the maximum decoding delay and the sum decoding delay, respectively, against the number of players $M$ for $N=30$, $Q=0.2$, and $P=0.1$. \fref{fig:PQMAX} and \fref{fig:PQCT} show this comparison in perfect feedback scenario  for the maximum decoding delay and the sum decoding delay, respectively, against the player-player erasure probability $P/Q$ while keeping the BS-players erasure probability $Q$ fixed for $M=60$, $N=30$ and  $Q=0.3$.

\fref{fig:MMAXL} and \fref{fig:MSUML} depict the aforementioned comparison in lossy feedback scenario against number of players $M$ for $N=30$, $Q=0.2$, and $P=0.1$ and \fref{fig:PQMAXL} and \fref{fig:PQCTL} show this comparison in lossy feedback scenario against the player-player erasure probability $P/Q$ while keeping the BS-players erasure probability $Q$ fixed for $M=60$, $N=30$ and  $Q=0.3$.

From all the figures, we can clearly see that our optimal cooperative data exchange algorithms outperform the traditional PMP approach. \fref{fig:MCT}, \fref{fig:MMAX}, and \fref{fig:MSUM} illustrate the gain in using a distributed algorithm when the player-player channel conditions is better than the BS-player channel ($P=0.5Q$).

For a small number of players, the PMP setting is close to our optimal approach. This can be explained by the fact that, for a small number of players, the probability that the union of the Has sets of all the players is equal to $\mathcal{N}$ is low. Thus, the BS carries a good portion of the recovery process until this condition occurs and then CDE can take place. This makes the overall performance of close to the absolute PMP scheme. However, the larger the number of players, there is a much bigger probability that the union of their Has sets is equal to $\mathcal{N}$ and thus the more CDE will be employed, the bigger the gap to the PMP performance. 

\fref{fig:PQCT} illustrates the mean completion time against the player-player erasure probability for a fixed BS-player erasure. In this configuration, for the same erasure probability, the PMP scheme outperform the optimal distributed approach. This can be explained first by the fact that the BS has all the packets whereas any one player has only a subset and thus, in general, it has better ability to form coding combinations that target more users than any single player. It can also be explained by the fact that in our optimal CDE approach, the approximation of the completion time using the decoding delay approach requires the erasure probability between the sender in large sense and the player. For the PMP scheme only the base station is transmitting and therefore this probability is fixed (from BS to players). However in the optimal CDE scheme, all players can transmit and we approximate the probability by the average erasure linking each player to all the other players. This approximation degrades the scheme. However as the channel linking players become better, CDE starts to outperform even if players send less efficient coding combinations and we clearly can see the difference between our scheme and the PMP one. We also note that the low complexity CDE outperforms the PMP scheme for a good player-player channel comparing to the BS-player channel.

From \fref{fig:MCTL} and \fref{fig:PQCTL}, we clearly can see that the CDE outperforms the PMP in the lossy feedback scenario for a good enough player-player channel ($P\approx 0.6Q$). However, the CDE scheme rapidly degrades with the number of uncertainties (player-player erasure probability). This can be explained by the fact that for a low erasure, the uncertainties are relatively rare events and therefore the system is nearly equivalent to the low complexity CDE. However, as the erasure increases, the uncertainties also increases and players are more likely to take less efficient decision since the estimation will degrades. The same analysis is applicable for the maximum decoding delay (\fref{fig:MMAXL} and \fref{fig:PQMAXL}) and the sum decoding delay (\fref{fig:MSUML} and \fref{fig:PQSUML}).

\fref{fig:PQMAX} presents the average maximum decoding delay for a fixed BS-player erasure probability against the player-player erasure probability. We can clearly see that as the player-player communications become more reliable (geocentrically close players), the gap between our optimal CDE scheme and the PMP scheme increases. Note that also for the same erasure probability between player-player and the BS-player, our scheme outperforms the PMP. This can be explained by the fact that the equality of erasure is only in the average sense. Therefore, for a fixed packet combination, the minimum expected delay that can be achieved by one of the players, that can form this combination, is less than the expected delay when the base station sends the combination. The same line of toughs can be applied to explain \fref{fig:PQSUM}. For the maximum and the sum decoding delay, the low complexity algorithm performs worst that the PMP scheme in all scenarios. This can be explained by the fact that in the low complexity CDE, the system is suffering from both the conventional and the cooperative decoding delay. Whereas, the optimal PMP and optimal CDE do no suffer from the cooperative decoding delay.

\section{Conclusion and Future Perspectives}\label{sec:concl}

In this paper, we formulated the problem of minimizing the delay of instantly decodable network coding for cooperative data exchange in decentralized wireless network as cooperative control game. We employed game theory as a tool to improve the distributed solution by overcoming the need for a central controller or additional signaling in the system. We modeled the session by self-interested players in a non-cooperative potential game. The utility functions were designed in such way that increasing individual payoff results in a collective behavior achieving both a desirable system performance and the Nash bargaining solution. We formulated three games to minimize each delay aspect in IDNC and improved these formulations to include punishment policy and achieve Nash bargaining solution. Through extensive simulations, our framework is tested against the best performance that could be found in the conventional point-to-multipoint (PMP) recovery process in numerous cases: first we simulated the problem with complete information. We, then, simulate with incomplete information and finally we tested it in lossy feedback scenario. Numerical results showed that our formulation with complete information largely outperforms the conventional PMP scheme in most practical scenarios and achieved a lower delay. The advantage of this formulation is that it can be easily extended. For example, the extension to the limited feedback scenario is straightforward. As future perspectives, this formulation can be extended to the case where not all the players are in the range of each other. In other words, the utility function will not be completely defined for the players. Moreover, the set of NE of the game will be different from the set derived in this paper. Another interesting research direction is the multicast cast with limited range. In this scenario, the packet demand of each player can differ and players are not all in the transmission range of each other. Finally, the case of imperfect feedback is another important and more practical extension.

\appendices

\numberwithin{equation}{section}

\section{Proof of Theorem 3}
Note that the cost function can be written as:
\begin{align}
\phi^\prime(\underline{a}_t,\underline{h}_t) = \phi^\prime(\underline{a}_{t-1},\underline{h}_{t-1}) + \xi(\underline{a}_{t},\underline{h}_t),
\end{align}
with:
\begin{align}
\xi(\underline{a}_{t},\underline{h}_t) = ||\underline{\mathcal{C}}(\underline{a}_t,\underline{h}_t)||_\infty - ||\underline{\mathcal{C}}(\underline{a}_{t-1},\underline{h}_{t-1})||_\infty .
\end{align}
Let $Q(t)$ be the set of players that can potentially increase the cost function at stage $t$ of the game. The mathematical definition of this set is the following:
\begin{align}
&Q(t) = \{ i \in \mathcal{M} \text{ such that }  \\
& \qquad \mathcal{C}_i(\underline{a}_{t-1},\underline{h}_{t-1}) + 1/(1-\overline{p}_i) > ||\underline{\mathcal{C}}(\underline{a}_{t-1},\underline{h}_{t-1})||_\infty \nonumber \\
& \hspace{2cm} \text{ and } M^w_i = 1\}. \nonumber
\end{align} 

Clearly, if $Q(t)= \varnothing$, then any action profile $\underline{a}_{t}^* \in \mathcal{A}(t)$ is a NE since all the profiles will not change the cost function $\phi^\prime(\underline{a}_t^*,\underline{h}_t) = \phi^\prime(\underline{a}_{t-1},\underline{h}_{t-1})$. In that case, we have $PoA(t) = 1$. Now assume $Q(t) \neq \varnothing$. For action profiles $\underline{a}_{t}$ the cost function will increase according to the norm of the action profile by the quantity:
\begin{align}
&\phi^\prime(\underline{a}_t,\underline{h}_t) - \phi^\prime(\underline{a}_{t-1},\underline{h}_{t-1}) = \nonumber \\
&\begin{cases}
\underset{i \in Q(t)}{max }\cfrac{1}{1-\overline{p}_i} \hspace{0.2cm} &\text{if } ||\underline{a}_t||_1 \neq 1 \\
\underset{i \in Q(t) \cap \underline{\mathcal{D}}_{w_j,\kappa^j}(t)}{max }\cfrac{1}{1-\overline{p}_i} \hspace{0.2cm} &\text{if } ||\underline{a}_t||_1 = a_j(t) = 1.
\end{cases}
\end{align}
Let $Z(t)$ be the set of players that can target players in the critical set $Q(t)$ and reduce the increase in the cost function. This set is defined as:
\begin{align}
Z(t) = \{ &j \in \mathcal{M} \text{ such that } \nonumber \\
& \underset{i \in Q(t) \cap \underline{\mathcal{D}}_{w_j,\kappa^j}(t)}{max }\cfrac{1}{1-\overline{p}_i} < \underset{i \in Q(t)}{max }\cfrac{1}{1-\overline{p}_i}\}.
\end{align}
Two cases can be distinguished:
\begin{itemize}
\item $Z(t) = \varnothing$ : all action profile $\underline{a}_t^*$ will yield the same cast function $\phi^\prime(\underline{a}_t^*,\underline{h}_t) = \phi^\prime(\underline{a}_{t-1},\underline{h}_{t-1}) + \underset{i \in Q(t)}{max }\cfrac{1}{1-\overline{p}_i}$. Therefore all action profile are NE of the game and the PoA is equal to $1$.
\item $Z(t) \neq 0$ : action profile $\underline{a}_t^*=a_j(t)$ such that $j\in Z(t)$ will yield a lower cost function than the other profiles.
\end{itemize}

Define $Y_0(t) = \underset{i \in Q(t)}{max }\cfrac{1}{1-\overline{p}_i}$ as the increase in the cost function when not exactly one player is transmitting and $Y_j(t) = \underset{i \in Q(t) \cap \underline{\mathcal{D}}_{w_j,\kappa^j}(t)}{max }\cfrac{1}{1-\overline{p}_i}$ the increase when only player $j$ is transmitting. Clearly for action profiles $\underline{a}_t^*$ such that $||\underline{a}_t^*||_1=0$ are not NE since the unilateral deviation of player $i \in Z(t)$ will decrease the cost function. For action profile $\underline{a}_t^*$ such that $||\underline{a}_t^*||_1>2$, for any unilateral deviation, we have $||\underline{a}_{t,i},\underline{a}_{t,-i}^*||>1$. Therefore, the cost function is unchanged and all these action profiles are NE. Let $\underline{a}_t^*$ be an action profile such that $||\underline{a}_t^*||_1=a_j(t)=1$, then the difference in the cost function for any unilateral deviation of player is:
\begin{align}
&\phi^\prime(\underline{a}_{t,i},\underline{a}_{t,-i}^,\underline{h}_{t}) - \phi^\prime(\underline{a}_t^*,\underline{h}_t) \nonumber \\
&= \xi(\underline{a}_{t,i},\underline{a}_{t,-i}^*,\underline{h}_t) - \xi(\underline{a}_{t}^*,\underline{h}_t) =Y_0(t) - Y_j(t).
\end{align} 

Thus, any unilateral deviation will yield the same or a higher cost function. These profiles are NE of the game. Let the action profile be $\underline{a}_{t}^*$ with $||\underline{a}_{t}^*||_1=a_i(t)+a_j(t)=2$. Two scenarios can occur:
\begin{itemize}
\item $i \notin Z(t)$ and $j \notin Z(t)$: By the same argument than for the case where $||\underline{a}_{t}^*||_1>2$, any unilateral deviation will yield the same cost $\phi^\prime(\underline{a}_t^*,\underline{h}_t) = \phi^\prime(\underline{a}_{t-1},\underline{h}_{t-1})+Y_0(t)$. Therefore it is a NE of the game.
\item $i \in Z(t)$ or $j \in Z(t)$: By the same argument than for the case $\underline{a}_{t}^*=\underline{0}$, the cost function can be decreased by unilateral deviation of player $i$ if $j \in Z(t)$ and inversely. This scenario is not a NE of the game.
\end{itemize}
We now characterize the set of all NE of the game:
\begin{align}
E(t) =  
\begin{cases}
\mathcal{A}(t) \text{ if } Z(t) = \varnothing \\
E_1(t)  \text{ otherwise },
\end{cases} 
\end{align}
where
\begin{align}
&E_1(t) = \{ \underline{a}_{t} \in \mathcal{A}(t)  \text{ such that } ||\underline{a}_{t}||_1=1 \text{ or } ||\underline{a}_{t}||_1>2  \nonumber\\
& \qquad \text{ or } (||\underline{a}_{t}||_1=a_i(t)+a_j(t)=2 \text{ and } i,j \notin Z(t))\}.
\end{align}
The \emph{PoA} of Game 3 can be expressed as follows:
\begin{align}
&PoA(t) = \nonumber \\
&\begin{cases}
1 \hspace{0.5cm} &\text{if } Z(t) = \varnothing \\
1 - \cfrac{Y_0(t) - \underset{j \in Z(t)}{min }(Y_j(t))}{\phi^\prime (\underline{a}_{t-1},\underline{h}_{t-1})+Y_0(t)} \hspace{0.5cm} &\text{otherwise }.
\end{cases}
\end{align}

\section{Proof of Theorem 4}
Note that the cost function can be written as:
\begin{align}
\phi^\prime(\underline{a}_t,\underline{h}_t) = \phi^\prime(\underline{a}_{t-1},\underline{h}_{t-1}) + \eta(\underline{a}_{t},,\underline{h}_t),
\end{align}
with:
\begin{align}
\eta(\underline{a}_{t},\underline{h}_t) = 
\begin{cases}
1 \hspace{0.1cm}& \text{if } ||\underline{\mathds{D}}(\underline{a}_t,\underline{h}_t)||_\infty >  ||\underline{\mathds{D}}(\underline{a}_{t-1},\underline{h}_{t-1})||_\infty \\
0 \hspace{0.1cm}& \text{otherwise } .
\end{cases}
\end{align}

Let $Q(t)$ be the set of players that can potentially increase the cost function at stage $t$ of the game. The mathematical definition of this set is the following:
\begin{align}
&Q(t) = \{ i \in \mathcal{M} \text{ such that }  \\
& \qquad \mathds{D}_i(\underline{a}_{t-1},\underline{h}_{t-1}) = ||\underline{\mathds{D}}(\underline{a}_{t-1},\underline{h}_{t-1})||_\infty \text{ and } M^w_i = 1\}. \nonumber
\end{align} 

Clearly, if $Q(t)= \varnothing$, then any action profile $\underline{a}_{t}^* \in \mathcal{A}(t)$ is a NE since all the profiles will not change the cost function $\phi^\prime(\underline{a}_t^*,\underline{h}_t) = \phi^\prime(\underline{a}_{t-1},\underline{h}_{t-1})$. In that case, we have $PoA(t) = 1$. Now assume $Q(t) \neq \varnothing$. Let $\underline{q}(t) \in \{0,1\}^M$ such that $q_i(t) = 1$ iff $i \in Q(t)$. Two cases can be distinguished:
\begin{enumerate}
\item $\nexists~ i \in \mathcal{M}$ such that $\underline{\mathcal{D}}_{w_i,\kappa^i}(t) \circ \underline{q}(t) = \underline{0}$.
\item $\exists~ i \in \mathcal{M}$ such that $\underline{\mathcal{D}}_{w_i,\kappa^i}(t) \circ \underline{q}(t) = \underline{0}$. Let $Z(t)$ be the set of such players.
\end{enumerate}

In the first case, all action profile will yield the same value of the cost $\phi^\prime(\underline{a}_t^*,\underline{h}_t) = \phi^\prime(\underline{a}_{t-1},\underline{h}_{t-1})+1$ and thus all of them are NE of the game and $PoA(t) = 1$. In the second case, the cost varies with the chosen action profile. Let $\underline{a}_{t}^*$ be an action profile such that $||\underline{a}_{t}^*||_1>2$. For any unilateral deviation, we have $||(\underline{a}_{t,i},\underline{a}_{t,-i}^*)||_1>1$. By definition of the delay $\underline{\mathds{D}}$, we have:
\begin{align}
\underline{\mathds{D}}(\underline{a}_{t,i},\underline{a}_{t,-i}^*,\underline{h}_t) = \underline{\mathds{D}}(\underline{a}_{t-1},\underline{h}_{t-1}) + \underline{M}^w(t).
\end{align}
Since we assumed $Q(t) \neq \varnothing$ then we have:
\begin{align}
||\underline{\mathds{D}}(\underline{a}_{t,i},\underline{a}_{t,-i}^*,\underline{h}_t)||_\infty &= ||\underline{\mathds{D}}(\underline{a}_{t-1},\underline{h}_{t-1})||_\infty + 1 \nonumber \\
&= ||\underline{\mathds{D}}(\underline{a}_{t}^*,\underline{h}_t)||_\infty.
\end{align}

Thus all action profile $\underline{a}_{t}^*$ with $||\underline{a}_{t}^*||_1>2$ are NE of the game that yields the maximum value of the cost $\phi^\prime(\underline{a}_t^*,\underline{h}_t) = \phi^\prime(\underline{a}_{t-1},\underline{h}_{t-1})+1$. Consider now the action profile $\underline{a}_{t}^*=\underline{0}$, then for some deviation the cost decreases by the quantity:
\begin{align}
&\phi^\prime(\underline{a}_{t}^*,\underline{h}_t) - \phi^\prime(\underline{a}_{t,i},\underline{a}_{t,-i}^*,\underline{h}_t) = 1 \nonumber\\
&,~\forall~\underline{a}_{t,i} \neq \underline{a}_{t,i}^* \in \mathcal{A}_i(t),~\forall~i \in Z(t) .
\end{align}
Therefore, the action profile of this type are not NE of the game. Let the action profile be $\underline{a}_{t}^*$ with $||\underline{a}_{t}^*||_1=a_i(t)+a_j(t)=2$. Two scenarios can occur:
\begin{itemize}
\item $i \notin Z(t)$ and $j \notin Z(t)$: By the same argument than for the case where $||\underline{a}_{t}^*||_1>2$, any unilateral deviation will yield the same cost $\phi^\prime(\underline{a}_t^*,\underline{h}_t) = \phi^\prime(\underline{a}_{t-1},\underline{h}_{t-1})+1$. Therefore it is a NE of the game.
\item $i \in Z(t)$ or $j \in Z(t)$: By the same argument than for the case $\underline{a}_{t}^*=\underline{0}$, the cost function can be decreased by unilateral deviation of player $i$ if $j \in Z(t)$ and inversely. This scenario is not a NE of the game.
\end{itemize}
Finally consider the action profile $\underline{a}_{t}^*$ with $||\underline{a}_{t}^*||_1=a_i(t)=1$. The following options may arise:
\begin{itemize}
\item $i \notin Z(t)$: By the same argument than for the case where $||\underline{a}_{t}^*||_1>2$, any unilateral deviation will yield the same cost $\phi^\prime(\underline{a}_t^*,\underline{h}_t) = \phi^\prime(\underline{a}_{t-1},\underline{h}_{t-1})+1$. Therefore it is a NE of the game.
\item $i \in Z(t)$: All deviation will increase the cost function. Hence it is a NE of the game.
\end{itemize}
We now characterize the set of all NE of the game:
\begin{align}
E(t) =  
\begin{cases}
\mathcal{A}(t) \text{ if } Z(t) = \varnothing \\
E_1(t)  \text{ otherwise },
\end{cases} 
\end{align}
where
\begin{align}
&E_1(t) = \{ \underline{a}_{t} \in \mathcal{A}(t)  \text{ such that } ||\underline{a}_{t}||_1=1 \text{ or } ||\underline{a}_{t}||_1>2  \nonumber\\
& \qquad \text{ or } (||\underline{a}_{t}||_1=a_i(t)+a_j(t)=2 \text{ and } i,j \notin Z(t))\}.
\end{align}
The \emph{PoA} of Game 2 can be expressed as follows:
\begin{align}
PoA(t) = 
\begin{cases}
1 \hspace{0.5cm} &\text{if } Z(t) = \varnothing \\
1 - \cfrac{1}{\phi^\prime (\underline{a}_{t-1},\underline{h}_{t-1})+1} \hspace{0.5cm} &\text{otherwise }.
\end{cases}
\end{align}

\section{Proof of Theorem 5}
Note that the cost function can be written as:
\begin{align}
\phi^\prime(\underline{a}_t,\underline{h}_t) = \phi^\prime(\underline{a}_{t-1},\underline{h}_{t-1}) + \psi(\underline{a}_{t},\boldsymbol{\omega}(t)),
\end{align}
with:
\begin{align}
&\psi((\underline{a}_{t},\boldsymbol{\omega}(t)) =  \\
&\begin{cases}
||\underline{M}^w(t)||_1 \hspace{0.1cm}& \text{if } ||\underline{a}_t||_1 \neq 1\\
||\underline{\mathcal{D}}_{w_i,\kappa^i}(t)||_1 \hspace{0.1cm}& \text{otherwise } \text{with } i \text{ such that } a_i(t) = ||\underline{a}_t||_1 .
\end{cases}\nonumber
\end{align}
By definition of $\underline{\mathcal{D}}_{w_i,\kappa^i}(t)$, we have $||\underline{M}^w(t)||_1 > ||\underline{\mathcal{D}}_{w_i,\kappa^i}(t)||_1,~\forall~i \in \mathcal{M}$ and therefore:
\begin{align}
\underset{\underline{a}_{t} \in \mathcal{A}(t)}{\text{argmax }}\phi^\prime(\underline{a}_t,\underline{h}_t) &= \underset{\underline{a}_{t} \in \mathcal{A}(t)}{\text{argmax }}\psi((\underline{a}_{t},\boldsymbol{\omega}(t))\nonumber \\
&= \underset{||\underline{a}_t||_1 \neq 1}{\text{argmax }} ||\underline{M}^w(t)||_1  \\
\underset{\underline{a}_{t} \in \mathcal{A}(t)}{\text{argmin }}\phi^\prime(\underline{a}_t,\underline{h}_t) &= \underset{\underline{a}_{t} \in \mathcal{A}(t)}{\text{argmin }} \psi((\underline{a}_{t},\boldsymbol{\omega}(t))  \nonumber \\
&=\underset{||\underline{a}_t||_1 =a_i(t)=1}{\text{argmin }} ||\underline{\mathcal{D}}_{w_i,\kappa^i}(t)||_1.
\end{align}

Let $\underline{a}_{t}^*$ be an action profile such that $||\underline{a}_{t}^*||_1>2$. For any unilateral deviation, we have $||(\underline{a}_{t,i},\underline{a}_{t,-i}^*)||_1>1$. Therefore the cost function is unchanged by any unilateral deviation i.e. $\phi(\underline{a}_{t,i},\underline{a}_{t,-i}^*,\underline{h}_t) = \phi(\underline{a}_{t}^*,\underline{h}_t) ,~\forall~\underline{a}_{t,i} \in \mathcal{A}_i(t),~\forall~i \in \mathcal{M}$. Hence all action profile $\underline{a}_{t}^*$ such that $||\underline{a}_{t}^*||_1>2$ are NE. These NE yields the highest cost $(\phi(\underline{a}_{t-1},\underline{h}_{t-1}) - ||\underline{M}^w(t)||_1)$ among all possible action profiles. Then, we have:
\begin{align}
\underset{s \in E(t)}{\text{max}}\phi^\prime(s,\underline{h}_t) = \phi^\prime(\underline{a}_{t-1},\underline{h}_{t-1}) + ||\underline{M}^w(t)||_1.
\end{align}

Similarly, let $\underline{a}_{t}^*$ be an action profile such that $||\underline{a}_{t}^*||_1=a_j(t)=1$. For any unilateral deviation, we have $||(\underline{a}_{t,i},\underline{a}_{t,-i}^*)||_1 \neq 1$. In terms of cost function, we have:
\begin{align}
&\phi^\prime(\underline{a}_{t,i},\underline{a}_{t,-i}^*,\underline{h}_t) = \phi^\prime(\underline{a}_{t-1},\underline{h}_{t-1}) + ||\underline{M}^w(t)||_1  \nonumber\\
& \quad > \phi^\prime(\underline{a}_{t}^*,\underline{h}_t) = \phi^\prime(\underline{a}_{t-1},\underline{h}_{t-1}) + ||\underline{\mathcal{D}}_{w_j,\kappa^j}(t)||_1 \nonumber\\
&,~\forall~\underline{a}_{t,i} \in \mathcal{A}_i(t),~\forall~i \in \mathcal{M}.
\end{align}
Which conclude that all action profile $\underline{a}_{t}^*$ such that $||\underline{a}_{t}^*||_1=1$ are NE. Let $\underline{a}_{t}^*=\underline{0}$, then any deviation decreases the cost by the quantity:
\begin{align}
&\phi^\prime(\underline{a}_{t}^*,\underline{h}_t) - \phi^\prime(\underline{a}_{t,i},\underline{a}_{t,-i}^*,\underline{h}_t) =||\underline{M}^w(t)||_1 - ||\underline{\mathcal{D}}_{w_i,\kappa^i}(t)||_1  \nonumber\\
&,~\forall~\underline{a}_{t,i} \neq \underline{a}_{t,i}^* \in \mathcal{A}_i(t),~\forall~i \in \mathcal{M}.
\label{eqqq}
\end{align}

Clearly, these action profiles are not NE of the game. Finally, consider the action profile $\underline{a}_{t}^*$ such that $||\underline{a}_{t}^*||_1=a_i(t)+a_j(t)=2$. The unilateral deviation of player $i$ or $j$ will decrease the cost by the same quantity than in equation \eref{eqqq} (replace $i$ by $j$ if it is player $j$ that deviates).Which conclude that these type of action profile are not NE of the game. We now can characteristic the set of all possible NE at stage $t$ of the game:
\begin{align}
E(t) = \{ \underline{a}_{t} \in \mathcal{A}(t) \text{ such that } ||\underline{a}_{t}||_1=1 \text{ or } ||\underline{a}_{t}||_1>2\}.
\end{align}
The \emph{PoA} of Game 1 can be expressed as follows:
\begin{align}
PoA(t) = \cfrac{\underset{||\underline{a}_t||_1=a_i(t)=1}{\text{min}}\phi^\prime (\underline{a}_{t-1},\underline{h}_{t-1})+||\underline{\mathcal{D}}_{\omega_i,\kappa^i}||_1}{\phi^\prime (\underline{a}_{t-1},\underline{h}_{t-1})+||\underline{M}^w(t)||_1}.
\end{align}
We can clearly see that $0 \leq ||\underline{\mathcal{D}}_{w_j,\kappa^j}(t)||_1 \leq ||\underline{M}^w(t)||_1-1$. Therefore, the \emph{PoA} can be bound by the following quantities:
\begin{align}
&1-\cfrac{||\underline{M}^w(t)||_1}{\phi^\prime(\underline{a}_{t-1},\underline{h}_{t-1})+||\underline{M}^w(t)||_1}    \leq \nonumber \\
& \qquad PoA(t) \leq 1-\cfrac{1}{\phi^\prime(\underline{a}_{t-1},\underline{h}_{t-1})+||\underline{M}^w(t)||_1}.
\end{align}

\section{Proof of Theorem 7}
As for game 3, the utility function of Game 4 can be written as:
\begin{align}
\phi^\prime(\underline{a}_t,\underline{h}_t) = \phi^\prime(\underline{a}_{t-1},\underline{h}_{t-1}) + \xi(\underline{a}_{t},\underline{h}_t),
\end{align}
with:
\begin{align}
\xi(\underline{a}_{t},\underline{h}_t) &=  ||\underline{a}_t||_1 + \cfrac{||\underline{\mathds{D}}(\underline{a}_t,\underline{h}_t)-\underline{\mathds{D}}(\underline{a}_{t-1},\underline{h}_{t-1})||_1}{M} \nonumber \\
& \quad +  ||\underline{\mathcal{C}}(\underline{a}_t,\underline{h}_t)||_\infty - ||\underline{\mathcal{C}}(\underline{a}_{t-1},\underline{h}_{t-1})||_\infty .
\end{align}
Let $Q(t)$ be the set of players that can potentially increase the expected completion time at stage $t$ of the game. The mathematical definition of this set is the following:
\begin{align}
&Q(t) = \{ i \in \mathcal{M} \text{ such that }  \\
& \qquad \mathcal{C}_i(\underline{a}_{t-1},\underline{h}_{t-1}) + 1/(1-\overline{p}_i) > ||\underline{\mathcal{C}}(\underline{a}_{t-1},\underline{h}_{t-1})||_\infty \nonumber \\
& \hspace{2cm} \text{ and } M^w_i = 1\}. \nonumber
\end{align} 
Assume $Q(t)= \varnothing$, then any action profile $\underline{a}_{t}^* \in \mathcal{A}(t)$ will not increase the maximum decoding delay and we have:
\begin{align}
\xi(\underline{a}_{t},\underline{h}_t) &=  ||\underline{a}_t||_1 + \cfrac{||\underline{\mathds{D}}(\underline{a}_t,\underline{h}_t)-\underline{\mathds{D}}(\underline{a}_{t-1},\underline{h}_{t-1})||_1}{M}  \\
=&\begin{cases}
\cfrac{||\underline{M}^w(t)||_1}{M} \hspace{0.1cm}& \text{if } ||\underline{a}_t||_1=0 \\
1 +  \cfrac{||\underline{\mathcal{D}}_{w_i,\kappa^i}(t)||_1}{M} \hspace{0.1cm}& \text{if } ||\underline{a}_t||_1=a_i(t)=1\\
||\underline{a}_t||_1 + \cfrac{||\underline{M}^w(t)||_1}{M} \hspace{0.1cm}& \text{otherwise } .
\end{cases} \nonumber
\end{align}

Clearly, we can see that all the profiles $\underline{a}_{t}^*$ with $||\underline{a}_{t}^*||_1>1$ or $||\underline{a}_{t}^*||_1=0$ are not anymore NE of the game. Only action profile with one entry not equal to $0$ are NE. Now assume $Q(t) \neq \varnothing$ and let $Z(t)$ be the set of players that can target players in the critical set $Q(t)$ and reduce the increase in the cost function. This set is defined as:
\begin{align}
Z(t) = \{ &j \in \mathcal{M} \text{ such that } \nonumber \\
& \underset{i \in Q(t) \cap \underline{\mathcal{D}}_{w_j,\kappa^j}(t)}{max }\cfrac{1}{1-\overline{p}_i} < \underset{i \in Q(t)}{max }\cfrac{1}{1-\overline{p}_i}\}.
\end{align}
Two cases can be distinguished:
\begin{itemize}
\item $Z(t) = \varnothing$ : the expected completion time will increase by the same amount $\underset{i \in Q(t)}{max }\cfrac{1}{1-\overline{p}_i}$ for all action  profiles.
\item $Z(t) \neq 0$ : some action profiles lead to lower increase of the expected completion time than others.
\end{itemize}

Define $Y_0(t) = \underset{i \in Q(t)}{max }\cfrac{1}{1-\overline{p}_i}$ as the increase in the cost function when not exactly one player is transmitting and $Y_j(t) = \underset{i \in Q(t) \cap \underline{\mathcal{D}}_{w_j,\kappa^j}(t)}{max }\cfrac{1}{1-\overline{p}_i}$ the increase when only player $j$ is transmitting.
If $Z(t) = \varnothing$ for an action profile $\underline{a}_{t}^*$, the cost function can be expressed as:
\begin{align}
&\xi(\underline{a}_{t}^*,\underline{h}_t) = Y_0(t) \nonumber \\
& +
\begin{cases}
\cfrac{||\underline{M}^w(t)||_1}{M} \hspace{0.1cm}& \text{if } ||\underline{a}_t||_1=0 \\
1 +  \cfrac{||\underline{\mathcal{D}}_{w_i,\kappa^i}(t)||_1}{M} \hspace{0.1cm}& \text{if } ||\underline{a}_t||_1=a_i(t)=1\\
||\underline{a}_t||_1 + \cfrac{||\underline{M}^w(t)||_1}{M} \hspace{0.1cm}& \text{otherwise } .
\end{cases}
\end{align}

Clearly action profiles $\underline{a}_{t}^*$ of type $||\underline{a}_{t}^*||_1 \neq 1$ are no more NE of the game. Now assume $Z(t) \neq \varnothing$ the utility varies with the chosen action profile. Let $\underline{a}_{t}^*$ be an action profile such that $||\underline{a}_{t}^*||_1=\alpha>1$. For some unilateral deviation of player that are transmitting, the cost function will decrease. In other words, we have:
\begin{align}
&\phi^\prime(\underline{a}_{t,i},\underline{a}^*_{t,-i},\underline{h}_t) - \phi^\prime(\underline{a}^*_{t},\underline{h}_t)  = \\
&\begin{cases}
1 > 0  \hspace{0.7cm}\text{if } ||\underline{a}_{t,i},\underline{a}^*_{t,-i}||_1 > 1 \\ 
Y_0(t)-Y_i(t) + 1 + \cfrac{||\underline{M}^w(t)|| - ||\underline{\mathcal{D}}_{w_i,\kappa^i}(t)||_1}{M}  > 0 \\
 \hspace{1.5cm}\text{if } ||\underline{a}_{t,i},\underline{a}^*_{t,-i}||_1 = 1 \text{ and } i \in Z(t) \\
1 + \cfrac{||\underline{M}^w(t)|| - ||\underline{\mathcal{D}}_{w_i,\kappa^i}(t)||_1}{M}  > 0 \\
\hspace{1.5cm}\text{if } ||\underline{a}_{t,i},\underline{a}^*_{t,-i}||_1 = 1 \text{ and } i \notin Z(t).
\end{cases}\nonumber
\end{align}

Thus, all action profile $\underline{a}_{t}^*$ with $||\underline{a}_{t}^*||_1>1$ are not NE and it is clear to see that action profile $\underline{a}_{t}^*$ such that $||\underline{a}_{t}^*||_1=0$ are also not NE of the game. Let $\underline{a}_{t}^*$  be an action profile such that $||\underline{a}_{t}^*||_1=a_i(t)=1$ and $i \notin Z(t)$. The difference in the cost when player $i$ deviates is:
\begin{align}
&\phi^\prime(\underline{a}^*_{t},\underline{h}_t) - \phi^\prime(\underline{a}_{t,i},\underline{a}^*_{t,-i},\underline{h}_t)  \nonumber \\
& \qquad = 1 + \cfrac{||\underline{\mathcal{D}}_{w_i,\kappa^i}(t)||_1 - ||\underline{M}^w(t)|| }{M} > 0.
\end{align}

Therefore such action profile is not a NE. Now consider the action profile $\underline{a}_{t}^*$  such that $||\underline{a}_{t}^*||_1=a_i(t)=1$ and $i \in Z(t)$. The difference in the utility if any player deviates:
\begin{align}
&\phi^\prime(\underline{a}_{t,i},\underline{a}^*_{t,-i},\underline{h}_t) - \phi^\prime(\underline{a}^*_{t},\underline{h}_t) =  \\
& \hspace{0.5cm} \begin{cases}
1 + Y_0(t)-Y_i(t) + \cfrac{ ||\underline{M}^w(t)|| - ||\underline{\mathcal{D}}_{w_i,\kappa^i}(t)||_1 }{M} > 0 \\
\hspace{1.5cm}\text{if } ||\underline{a}_{t,j},\underline{a}^*_{t,-j}||_1 = 2\\
Y_0(t)-Y_i(t) + \cfrac{ ||\underline{M}^w(t)|| - ||\underline{\mathcal{D}}_{w_i,\kappa^i}(t)||_1}{M} > 0  \\
\hspace{1.5cm} \text{if } ||\underline{a}_{t,j},\underline{a}^*_{t,-j}||_1 = 0.
\end{cases} \nonumber
\end{align}
The set of all NE of the Game 4 can be expressed as:
\begin{align}
&E(t) =  \\
&\begin{cases}
\{ \underline{a}_{t} \in \mathcal{A}(t)  \text{ s.t. } ||\underline{a}_{t}||_1=a_i(t)=1,i \in Z(t) \} \\
\hspace{6cm} \text{ if } Z(t) \neq \varnothing \\
\{ \underline{a}_{t} \in \mathcal{A}(t)  \text{ s.t. } ||\underline{a}_{t}||_1=1 \} \\
\hspace{6cm}\text{ otherwise}.
\end{cases}  \nonumber
\end{align}
The \emph{PoA} of Game 4 can be expressed as follows:
\begin{align}
&PoA^\prime(t) = \\
&\begin{cases}
\cfrac{\underset{||\underline{a}_t||_1=a_i(t)=1}{\text{min}}\phi^\prime (\underline{a}_{t-1},\underline{h}_{t-1})+\cfrac{||\underline{\mathcal{D}}_{\omega_i,\kappa^i}||_1}{M} +1 + Y_i(t) }{\underset{||\underline{a}_t||_1=a_i(t)=1}{\text{max}}\phi^\prime (\underline{a}_{t-1},\underline{h}_{t-1})+\cfrac{||\underline{\mathcal{D}}_{\omega_i,\kappa^i}||_1}{M} +1 + Y_i(t)}  \\
\hspace{6cm} \text{ if } Z(t) \neq \varnothing \\
\cfrac{\underset{||\underline{a}_t||_1=a_i(t)=1}{\text{min}}\phi^\prime (\underline{a}_{t-1},\underline{h}_{t-1})+\cfrac{||\underline{\mathcal{D}}_{\omega_i,\kappa^i}||_1}{M}  +1 + Y_0(t) }{\underset{||\underline{a}_t||_1=a_i(t)=1}{\text{max}}\phi^\prime (\underline{a}_{t-1},\underline{h}_{t-1})+\cfrac{||\underline{\mathcal{D}}_{\omega_i,\kappa^i}||_1}{M} +1 + Y_0(t)} \\
\hspace{6cm} \text{ otherwise }.
\end{cases} \nonumber
\end{align}
Moreover, the \emph{PoA} can be bounded by the following expressions:
\begin{align}
&1 \geq PoA^\prime(t) \geq 
\begin{cases}
1 - \cfrac{1+Y_0(t)-\underset{j \in Z(t)}{min }(Y_j(t))}{\phi^\prime (\underline{a}_{t-1},\underline{h}_{t-1})+2+Y_0(t)}  \hspace{0.1cm}  \\
\hspace{4cm} \text{if } Z(t) \neq \varnothing\\
1 - \cfrac{1}{\phi^\prime (\underline{a}_{t-1},\underline{h}_{t-1})+2+Y_0(t)} \hspace{0.1cm} \\
\hspace{4cm}\text{otherwise }.\\
\end{cases}
\end{align}

\section{Proof of Theorem 8}
As for game 2, the utility function of Game 4 can be written as:
\begin{align}
\phi^\prime(\underline{a}_t,\underline{h}_t) = \phi^\prime(\underline{a}_{t-1},\underline{h}_{t-1}) + \eta(\underline{a}_{t},,\underline{h}_t),
\end{align}
with:
\begin{align}
&\eta(\underline{a}_{t},\underline{h}_t) =  ||\underline{a}_t||_1 + \cfrac{||\underline{\mathds{D}}(\underline{a}_t,\underline{h}_t)-\underline{\mathds{D}}(\underline{a}_{t-1},\underline{h}_{t-1})||_1}{M} + \nonumber \\ 
&\begin{cases}
1 \hspace{0.1cm}& \text{if } ||\underline{\mathds{D}}(\underline{a}_t,\underline{h}_t)||_\infty >  ||\underline{\mathds{D}}(\underline{a}_{t-1},\underline{h}_{t-1})||_\infty \\
0 \hspace{0.1cm}& \text{otherwise } .
\end{cases}
\end{align}

Let $Q(t)$ be the set of players that can potentially increase the cost function at stage $t$ of the game as defined for Game 2. If $Q(t)= \varnothing$, then any action profile $\underline{a}_{t}^* \in \mathcal{A}(t)$ will not increase the maximum decoding delay and we have:
\begin{align}
\eta(\underline{a}_{t},\underline{h}_t) &=  ||\underline{a}_t||_1 + \cfrac{||\underline{\mathds{D}}(\underline{a}_t,\underline{h}_t)-\underline{\mathds{D}}(\underline{a}_{t-1},\underline{h}_{t-1})||_1}{M}  \\
=&\begin{cases}
\cfrac{||\underline{M}^w(t)||_1}{M} \hspace{0.1cm}& \text{if } ||\underline{a}_t||_1=0 \\
1 +  \cfrac{||\underline{\mathcal{D}}_{w_i,\kappa^i}(t)||_1}{M} \hspace{0.1cm}& \text{if } ||\underline{a}_t||_1=a_i(t)=1\\
||\underline{a}_t||_1 + \cfrac{||\underline{M}^w(t)||_1}{M} \hspace{0.1cm}& \text{otherwise } .
\end{cases} \nonumber
\end{align}

Clearly, we can see that all the profiles $\underline{a}_{t}^*$ with $||\underline{a}_{t}^*||_1>1$ or $||\underline{a}_{t}^*||_1=0$ are not anymore NE of the game. Only action profile with one entry not equal to $0$ are NE. Define $Z(t)$ as the set of players that can transmit a combination such that the maximum decoding delay do not increase (see definition in Game 2). If $Z(t) = \varnothing$ it is clear to see that only the action profile with $||\underline{a}_{t}^*||_1=a_i(t)=1$ are NE and yield the cost $\phi^\prime(\underline{a}_{t-1},\underline{h}_{t-1}) + 2 +  \cfrac{||\underline{\mathcal{D}}_{w_i,\kappa^i}(t)||_1}{M}$.

Now assume $Z(t) \neq \varnothing$ the utility varies with the chosen action profile. Let $\underline{a}_{t}^*$ be an action profile such that $||\underline{a}_{t}^*||_1=\alpha>1$. For some unilateral deviation of player that are transmitting, the cost function will decrease. In other words, we have:
\begin{align}
&\phi^\prime(\underline{a}^*_{t},\underline{h}_t) - \phi^\prime(\underline{a}_{t,i},\underline{a}^*_{t,-i},\underline{h}_t) = \\
&\begin{cases}
1 > 0  \hspace{0.7cm}\text{if } ||\underline{a}_{t,i},\underline{a}^*_{t,-i}||_1 > 1 \\ 
2 + \cfrac{||\underline{M}^w(t)|| - ||\underline{\mathcal{D}}_{w_i,\kappa^i}(t)||_1}{M}  > 0 \\
 \hspace{1.5cm}\text{if } ||\underline{a}_{t,i},\underline{a}^*_{t,-i}||_1 = 1 \text{ and } i \in Z(t) \\
1 + \cfrac{||\underline{M}^w(t)|| - ||\underline{\mathcal{D}}_{w_i,\kappa^i}(t)||_1}{M}  > 0 \\
\hspace{1.5cm}\text{if } ||\underline{a}_{t,i},\underline{a}^*_{t,-i}||_1 = 1 \text{ and } i \notin Z(t).
\end{cases}\nonumber
\end{align}

Thus, all action profile $\underline{a}_{t}^*$ with $||\underline{a}_{t}^*||_1>1$ are not NE and it is clear to see that action profile $\underline{a}_{t}^*$ such that $||\underline{a}_{t}^*||_1=0$ are also not NE of the game. Let $\underline{a}_{t}^*$  be an action profile such that $||\underline{a}_{t}^*||_1=a_i(t)=1$ and $i \notin Z(t)$. The difference in the cost when player $i$ deviates is:
\begin{align}
&\phi^\prime(\underline{a}^*_{t},\underline{h}_t) - \phi^\prime(\underline{a}_{t,i},\underline{a}^*_{t,-i},\underline{h}_t)  \nonumber \\
& \qquad = 1 + \cfrac{||\underline{\mathcal{D}}_{w_i,\kappa^i}(t)||_1 - ||\underline{M}^w(t)||}{M} > 0.
\end{align}

Therefore such action profile is not a NE. Now consider the action profile $\underline{a}_{t}^*$  such that $||\underline{a}_{t}^*||_1=a_i(t)=1$ and $i \in Z(t)$. The difference in the utility if any player deviates:
\begin{align}
&\phi^\prime(\underline{a}_{t,i},\underline{a}^*_{t,-i},\underline{h}_t) - \phi^\prime(\underline{a}^*_{t},\underline{h}_t) =  \\
& \hspace{1.5cm} \begin{cases}
2 + \cfrac{ ||\underline{M}^w(t)|| - ||\underline{\mathcal{D}}_{w_i,\kappa^i}(t)||_1 }{M} > 0 \\
\hspace{1.5cm} \text{if } ||\underline{a}_{t,j},\underline{a}^*_{t,-j}||_1 = 0\\
\cfrac{ ||\underline{M}^w(t)|| - ||\underline{\mathcal{D}}_{w_i,\kappa^i}(t)||_1}{M} > 0  \\
\hspace{1.5cm}\text{if } ||\underline{a}_{t,j},\underline{a}^*_{t,-j}||_1 = 2.
\end{cases} \nonumber
\end{align}
The set of all NE of the Game 4 can be expressed as:
\begin{align}
&E(t) =  \\
&\begin{cases}
\{ \underline{a}_{t} \in \mathcal{A}(t)  \text{ s.t. } ||\underline{a}_{t}||_1=1 \}\\
\hspace{6cm} \text{if } Z(t) = \varnothing \\
\{ \underline{a}_{t} \in \mathcal{A}(t)  \text{ s.t. } ||\underline{a}_{t}||_1=a_i(t)=1,i \in Z(t) \}\\
\hspace{6cm} \text{ otherwise}. 
\end{cases}  \nonumber
\end{align}
The \emph{PoA} of Game 4 can be expressed as follows:
\begin{align}
&PoA^\prime(t) = \\
&\begin{cases}
\cfrac{\underset{||\underline{a}_t||_1=a_i(t)=1}{\text{min}}\phi^\prime (\underline{a}_{t-1},\underline{h}_{t-1})+\cfrac{||\underline{\mathcal{D}}_{\omega_i,\kappa^i}||_1}{M} +2 }{\underset{||\underline{a}_t||_1=a_i(t)=1}{\text{max}}\phi^\prime (\underline{a}_{t-1},\underline{h}_{t-1})+\cfrac{||\underline{\mathcal{D}}_{\omega_i,\kappa^i}||_1}{M} +2} \\
\hspace{6cm} \text{ if } Z(t) = \varnothing \\
\cfrac{\underset{||\underline{a}_t||_1=a_i(t)=1}{\text{min}}\phi^\prime (\underline{a}_{t-1},\underline{h}_{t-1})+\cfrac{||\underline{\mathcal{D}}_{\omega_i,\kappa^i}||_1}{M} +1 }{\underset{||\underline{a}_t||_1=a_i(t)=1}{\text{max}}\phi^\prime (\underline{a}_{t-1},\underline{h}_{t-1})+\cfrac{||\underline{\mathcal{D}}_{\omega_i,\kappa^i}||_1}{M} +1}  \\
\hspace{6cm}\text{ otherwise }.
\end{cases} \nonumber
\end{align}
Moreover, the \emph{PoA} can be bounded by the following expressions:
\begin{align}
&1 \geq PoA^\prime(t) \geq 
\begin{cases}
1 - \cfrac{1}{\phi^\prime (\underline{a}_{t-1},\underline{h}_{t-1})+3}  \hspace{0.1cm} &\text{if } Z(t) = \varnothing\\
1 - \cfrac{1}{\phi^\prime (\underline{a}_{t-1},\underline{h}_{t-1})+2} \hspace{0.1cm} &\text{otherwise }.\\
\end{cases}
\end{align}

\section{Proof of Theorem 9}
Similarly to Game 1, the cost function of Game 3 can be written as:
\begin{align}
\phi^\prime(\underline{a}_t,\underline{h}_t) = \phi^\prime(\underline{a}_{t-1},\underline{h}_{t-1}) + \psi(\underline{a}_{t},\boldsymbol{\omega}(t)),
\end{align}
with:
\begin{align}
&\psi((\underline{a}_{t},\boldsymbol{\omega}(t)) =  \\
&\begin{cases}
||\underline{M}^w(t)||_1+ ||\underline{a}_t||_1 \hspace{0.1cm}& \text{if } ||\underline{a}_t||_1 \neq 1\\
||\underline{\mathcal{D}}_{w_i,\kappa^i}(t)||_1+ 1 \hspace{0.1cm}& \text{otherwise } \text{with } i \text{ s.t. } a_i(t) = ||\underline{a}_t||_1 .
\end{cases}\nonumber
\end{align}

Let $\underline{a}_{t}^*$ be an action profile such that $||\underline{a}_{t}^*||_1>2$. For any unilateral deviation of a player who is transmitting, the cost function increase. In other words, we have:
\begin{align}
&\phi^\prime(\underline{a}^*_{t},\underline{h}_t)- \phi^\prime(\underline{a}_{t,i},\underline{a}^*_{t,-i},\underline{h}_t)  =  \nonumber \\ 
&\hspace{4cm} ||\underline{a}_t^*||_1 - ||(\underline{a}_{t,i},\underline{a}^*_{t,-i})||_1 = 1 \nonumber \\
& \hspace{1cm}\forall~a_i(t) \neq a^*_i(t) \in \mathcal{A}_i(t) \text{ such that } a^*_i(t)=1.
\end{align}

The cost function changes by some unilateral deviation. Hence all action profile $\underline{a}_{t}^*$ such that $||\underline{a}_{t}^*||_1>2$ are no more NE of the game. By the same argument used in Game 1, the action profiles $\underline{a}_{t}^*$ such that $||\underline{a}_{t}^*||_1=0$ or $||\underline{a}_{t}^*||_1=2$ are not NE of the game.

For the action profile $\underline{a}_{t}^*$ such that $||\underline{a}_{t}^*||_1=a_j(t)=1$, the difference of cost function for any unilateral deviation is the following:
\begin{align}
&\phi^\prime(\underline{a}_{t,i},\underline{a}^*_{t,-i},\underline{h}_t) - \phi^\prime(\underline{a}^*_{t},\underline{h}_t) = \nonumber \\
&\begin{cases}
||\underline{M}^w(t)||_1 - ||\underline{\mathcal{D}}_{w_j,\kappa^j}(t)||_1 + 1 > 0 \\
\hspace{1cm}\forall~a_i(t) \neq a^*_i(t) \in \mathcal{A}_i(t),~\forall~i \in \mathcal{M} \setminus \{j\} \\
||\underline{M}^w(t)||_1  - ||\underline{\mathcal{D}}_{w_j,\kappa^j}(t)||_1 - 1 \geq 0  \\
\hspace{1cm}\forall~a_j(t) \neq a^*_j(t) \in \mathcal{A}_j(t).
\end{cases}
\end{align}

Since no unilateral deviation can yield a lower cost, therefore all action profile $\underline{a}_{t}^*$ such that $||\underline{a}_{t}^*||_1=1$ are NE of the game. The set of all possible NE at stage $t$ of the game becomes:
\begin{align}
E(t) = \{ \underline{a}_{t} \in \mathcal{A}(t) \text{ such that } ||\underline{a}_{t}||_1=1 \}.
\end{align}
The \emph{PoA} of Game 3 can be expressed as:
\begin{align}
PoA(t) = \cfrac{\underset{||\underline{a}_t||_1=a_i(t)=1}{\text{min}}\phi^\prime (\underline{a}_{t-1},\underline{h}_{t-1})+||\underline{\mathcal{D}}_{\omega_i,\kappa^i}||_1 +1 }{\underset{||\underline{a}_t||_1=a_i(t)=1}{\text{max}}\phi^\prime (\underline{a}_{t-1},\underline{h}_{t-1})+||\underline{\mathcal{D}}_{\omega_i,\kappa^i}||_1 +1}.
\end{align}
The \emph{PoA} of Game 3 can be bound by:
\begin{align}
& 1 - \cfrac{||\underline{M}^w(t)||_1-1}{\phi^\prime(\underline{a}_{t-1},\underline{h}_{t-1})+||\underline{M}^w(t)||_1} \leq PoA(t) \leq 1 .
\end{align}

\bibliographystyle{IEEEtran}
\bibliography{references}

\end{document}